\documentclass[11pt]{article}

\usepackage{algorithm}
\usepackage[noend]{algorithmic}

\usepackage{thmtools}
\usepackage{thm-restate}

\usepackage{bbm}
\usepackage{tikz, tikz-3dplot}
\usepackage{wrapfig}

\usepackage{amsfonts}
\usepackage{amsmath}
\usepackage{amsthm}
\usepackage[utf8]{inputenc}
\usepackage[margin=1in]{geometry}
\usepackage[numbers]{natbib}
\usepackage{tikz, tikz-3dplot}
\usepackage{wrapfig}

\newcommand{\calM}{\ensuremath{\mathcal{M}}}

\newcommand{\calO}{\ensuremath{\mathcal{O}}}

\newcommand{\calX}{\ensuremath{\mathcal{X}}}

\newtheorem{lemma}{Lemma}[section]
\newtheorem{theorem}[lemma]{Theorem}

\newtheorem{problem}[lemma]{Problem}
\newtheorem{definition}[lemma]{Definition}

\newtheorem{assumption}[lemma]{Assumption}
\newtheorem{claim}[lemma]{Claim}

\newcommand{\betad}[2]{\mathsf{Beta\left(#1, #2\right)}}

\newcommand{\countball}{B_{\mathsf{count}}}

\newcommand{\E}[2]{\mathbb{E}_{#1}\left[ #2 \right]}
\newcommand{\enc}[1]{\mathsf{Enc\left(#1\right)}}

\newcommand{\eps}{\varepsilon}

\newcommand{\gammad}[2]{\mathsf{Gamma}\left(#1, #2\right)}

\newcommand{\indic}[1]{\mathbbm{1}_{#1}}

\newcommand{\ord}{\mbox{ord}}
\renewcommand{\P}[2]{\mathbb{P}_{#1}\left[#2\right]}

\newcommand{\sumball}{B_{\mathsf{sum}}}

\newcommand{\voteball}{B_{\mathsf{vote}}}

\newcommand{\arxiv}[1]{}

\usepackage[colorlinks, citecolor=blue, linkcolor=blue, urlcolor=blue]{hyperref}
\usepackage{cleveref} 

\title{Some Constructions of Private, Efficient, and Optimal \\ $K$-Norm and Elliptic Gaussian Noise}
\author{Matthew Joseph\thanks{mtjoseph@google.com. Google Research.} \and Alexander Yu\thanks{alexjyu@google.com. Google Research.}}

\begin{document}

\maketitle

\begin{abstract}
    Differentially private computation often begins with a bound on some $d$-dimensional statistic's $\ell_p$ sensitivity. For pure differential privacy, the $K$-norm mechanism can improve on this approach using a norm tailored to the statistic's sensitivity space. Writing down a closed-form description of this optimal norm is often straightforward. However, running the $K$-norm mechanism reduces to uniformly sampling the norm's unit ball; this ball is a $d$-dimensional convex body, so general sampling algorithms can be slow. Turning to concentrated differential privacy, elliptic Gaussian noise offers similar improvement over spherical Gaussian noise. Once the shape of this ellipse is determined, sampling is easy; however, identifying the best such shape may be hard.
    
    This paper solves both problems for the simple statistics of sum, count, and vote. For each statistic, we provide a sampler for the optimal $K$-norm mechanism that runs in time $\tilde O(d^2)$ and derive a closed-form expression for the optimal shape of elliptic Gaussian noise. The resulting algorithms all yield meaningful accuracy improvements while remaining fast and simple enough to be practical. More broadly, we suggest that problem-specific sensitivity space analysis may be an overlooked tool for private additive noise.
\end{abstract}

\section{Introduction}
\label{sec:intro}
The Laplace mechanism~\citep{DMNS06} is a canonical method for computing pure differentially private (DP) statistics. \citet{HT10} showed that it can be viewed as the $K$-norm mechanism, which takes an input database $X$ and privately computes a $d$-dimensional statistic $T$ with $\|\cdot\|$-sensitivity $\Delta$ by outputting a draw from the density $f_X(y) \propto \exp\left(-\frac{\eps}{\Delta} \cdot \|y - T(X)\|\right)$, instantiated with the $\ell_1$ norm.~\citet{AS21} studied the choice of the optimal norm for $T$ and showed that it is uniquely determined by $T$'s sensitivity space, $S(T) = \{T(X) - T(X') \in \mathbb{R}^d \mid X, X' \text{ are neighbors}\}$. If the convex hull of $S(T)$ induces a norm, then it is the optimal norm.

Once a norm has been selected,~\citet{HT10} showed that sampling the $K$-norm mechanism reduces to uniformly sampling the norm unit ball and gave a black-box application of general results for sampling convex bodies. However, repeating this analysis with recent faster samplers tailored to convex polytopes~\citep{LLV20} only improves its arithmetic complexity to $\tilde O(d^{3+\omega})$ ($\omega \geq 2$ is the matrix multiplication exponent; see \Cref{subsec:prelims_k_norm} for details). Sampling the $K$-norm mechanism is therefore impractical for all but the smallest problems.

Turning to concentrated DP, a standard approach is to add spherical Gaussian noise calibrated to a statistic's $\ell_2$ sensitivity. Less coarsely, elliptic Gaussian noise~\citep{NTZ13} tailored to the statistic's sensitivity space is nearly instance optimal~\citep{NT23}. Sampling the noise is easy once its shape has been determined, but determining the best shape reduces to finding the minimum ellipse containing the sensitivity space. The general solution for this problem solves a semidefinite program~\citep{ENU20, NT23} for each $d$ and is only known to be approximately optimal in poly$(d)$ time for certain restricted classes of polytopes. Moreover, even for these classes, the polynomial has an impractically large degree (see \Cref{subsec:prelims_ellipse} for details).

\subsection{Contributions}
\label{subsec:contributions}
We consider three realistic problems: Sum, Count, and Vote. Short descriptions of these problems and results appear below. Throughout, the overall statistic $T$ is simply a linear query over points in the database, but the different assumptions about the data yield different sampling problems.

\begin{problem}[Sum]
\label{problem:sum}
    Each data point $x_i \in \mathbb{R}^d$ has $\|x_i\|_0 \leq k$ and $\|x_i\|_\infty \leq b$, i.e., each user contributes to at most $k$ quantities, and affects each by at most $b$. Systems employed by Google~\citep{WZLDS+20, AGJKV23} and LinkedIn~\citep{RSPDL+20} rely on similar ``contribution bounding'' to compute user-level private statistics.
\end{problem}

\begin{problem}[Count]
\label{problem:count}
    This is Sum with an additional nonnegativity constraint. It includes the histogram and top-$k$ problems used as running examples in the papers referenced in \Cref{problem:sum}.
\end{problem}

\begin{problem}[Vote]
\label{problem:vote}
    Each vector $x_i$ is a permutation of $(0, 1, \ldots, d-1)$. This encodes a setting where users rank $d$ options, and ranks are summed across users to vote. This process is used in several real-world voting systems~\citep{F14, B23}.
\end{problem}

All three problems have sensitivity spaces that yield non-$\ell_p$ optimal norm balls. Our first contribution is constructing efficient samplers for each one. This suffices to efficiently implement the optimal $K$-norm mechanisms (see \Cref{subsec:prelims_k_norm}). We also show that rejection sampling these norm balls is inefficient.

\begin{theorem}[Informal]
\label{thm:informal}
    The optimal $K$-norm mechanisms for Sum, Count, and Vote can be sampled in time $O(d^2)$, $O(d^2\log(d))$, and $O(d^2\log(d))$, respectively. Moreover, for any $p \in [1,\infty]$, rejection sampling any norm ball by sampling the $\ell_p$ ball takes time exponential in $d$.
\end{theorem}

The Sum ball is identical across orthants, so spherical Gaussian noise is optimal. For Count and Vote, our second contribution is deriving closed-form expressions for optimal elliptic Gaussian noise. The result for Count applies only in the sparse-contribution ($k \leq d/2$) setting, while the result for Vote is unrestricted.

\begin{theorem}[Informal]
\label{thm:vote_ellipse_informal}
    The enclosing ellipses for the sparse-contribution Count and Vote norm balls that minimize expected squared $\ell_2$ norm have closed forms and can be sampled in time $O(1)$.
\end{theorem}

Simulations (\Cref{fig:simulations}) show that the five algorithms yield nontrivial error improvements. Based on these results, the primary conceptual message of this paper is that problem-specific sensitivity space analysis is ``worth it'' to obtain practical algorithms.

\begin{figure}[t!]
    \centering
    \includegraphics[scale=0.5]{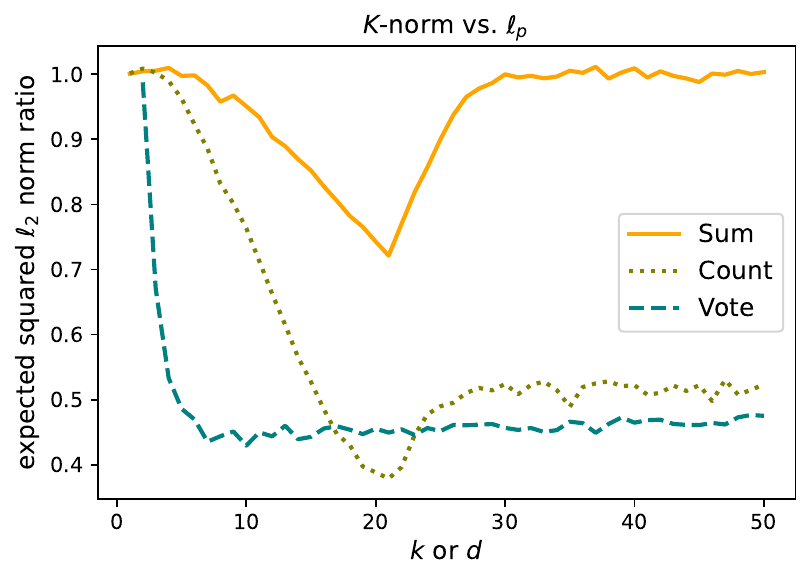}
    \includegraphics[scale=0.5]{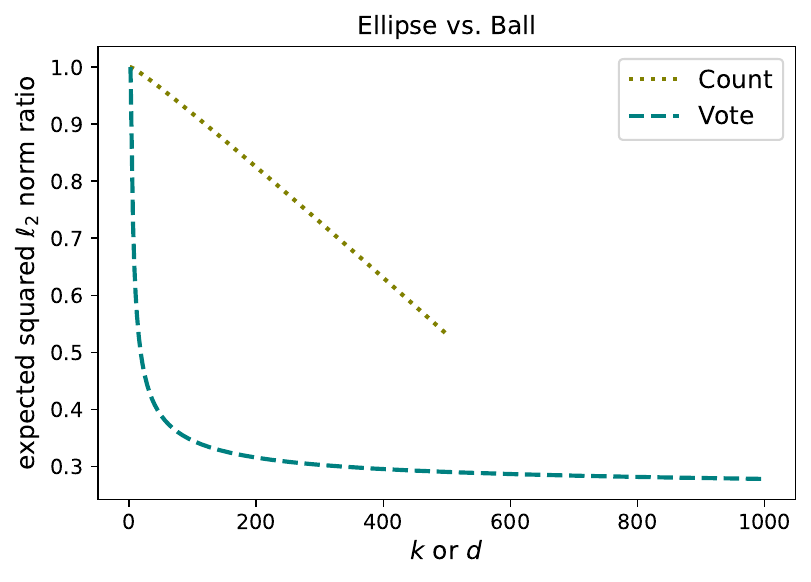}
    \caption{Mean squared $\ell_2$ error ratios. The privacy parameter $\eps$ or $\rho$ controls the scaling of a sample from the induced norm ball ($K$-norm mechanism) or ellipse (elliptic Gaussian noise), so we simply compare expected sample magnitudes for the underlying shapes. For the $K$-norm mechanism (left), we evaluate Sum and Count with dimension $d=50$ and varying contribution bound $k$. We also evaluate Vote, varying $d$ up to $d=50$ (note that Vote does not have a $k$ parameter). Each point compares to the best $\ell_p$ ball at the current parameter over 1,000 trials. For elliptic Gaussian noise (right), we compare to the minimum $\ell_2$ ball, fixing $d=1,000$ and varying $k$ for Count and varying $d$ up to $d=1,000$ for Vote, using closed-form expressions for the expected squared $\ell_2$ norm of a sample from the ellipse or ball in question. The Count ellipse plot covers $k \leq d/2$ because its minimal ellipse result only holds for this sparse-contribution setting. Throughout, a value $< 1$ means our algorithm is better. See Github~\cite{G24} for simulation code.}
    \label{fig:simulations}
\end{figure}

\subsection{Related Work}
\label{subsec:related_work}
Previous work gave efficient samplers for the $K$-norm mechanism using $\ell_2$~\citep{YRUF14} and $\ell_\infty$~\citep{SU16} norms, and efficiently sampling general $\ell_p$ balls reduces to sampling exponential and generalized gamma distributions~\citep{BGMN05}. \citet{HT10} and \citet{BDKT12} introduced better variants of the $K$-norm mechanism when the norm ball is far from isotropic position. However, the former's recursive algorithm relies on repeated estimation of the covariance matrices associated with ``smaller'' versions of the original norm ball, requiring $O(d^4)$ norm ball samples in total. The latter's algorithm requires sampling a randomly perturbed convex body, which falls back on the $O(d^{3+\omega})$ complexity for sampling a general convex body.

A similar line of work has studied private query answering. A common general strategy transforms a collection of queries, privately answers the new queries with oblivious (and typically Laplace or Gaussian) noise, and then translates the results back to the original collection. Solutions in this class include projection~\citep{NTZ13, N23B}, matrix~\citep{LMHMR15, MMHM18}, and factorization~\citep{ENU20, NT23} mechanisms. Instead of computing a better workload of queries to answer with a standard noise distribution, our application of the $K$-norm mechanism instead focuses on answering a single query with a non-standard noise distribution. Our derivations of elliptic Gaussian noise may be viewed as exact, efficient solutions for the optimal workload.

Finally, Vote has been studied in the context of private ranking~\citep{HEM17, AGKM22}. The nonadaptive algorithms in both works are improved by replacing their Laplace and Gaussian noise distributions with our $K$-norm and elliptic Gaussian noise.
\section{Preliminaries}
\label{sec:prelims}
We start with preliminaries from differential privacy. We use both pure and concentrated differential privacy, in the add-remove model.

\begin{definition}[\citet{DMNS06, BS16}]
\label{def:dp}
    Databases $X,X'$ from data domain $\calX$ are \emph{neighbors} $X \sim X'$ if they differ in the presence or absence of a single record. A randomized mechanism $\calM:\calX \to \calO$ is \emph{$\eps$-differentially private} (DP) if for all $X \sim X'\in \calX$ and any $S\subseteq \calO$, $\P{\calM}{\calM(X) \in S} \leq e^{\eps}\P{\calM}{\calM(X') \in S}$. Letting $D_\alpha$ denote $\alpha$-Renyi divergence, a randomized mechanism $\calM:\calX \to \calO$ is \emph{$\rho$-(zero) concentrated differentially private} (CDP) if for all $X \sim X'\in \calX$ and all $\alpha > 1$, $D_\alpha(M(X) \|\| M(X')) \leq \rho \alpha$.
\end{definition}

\subsection{$K$-Norm Mechanism}
\label{subsec:prelims_k_norm}

\begin{lemma}[\citet{HT10}]
\label{lem:k_norm_dp}
    Given statistic $T$ with $\|\cdot\|$-sensitivity $\Delta$ and database $X$, the \emph{$K$-norm mechanism} has output density $f_X(y) \propto \exp\left(-\frac{\eps}{\Delta} \cdot \|y - T(X)\|\right)$ and satisfies $\eps$-DP.
\end{lemma}

\begin{lemma}[Remark 4.2 \citet{HT10}]
\label{lem:k_norm_sample}
    The following procedure outputs a sample from the $K$-norm mechanism with norm $\|\cdot\|$, norm unit ball $B^d$, statistic $T(X)$, and statistic sensitivity $\Delta = 1$ with respect to $\|\cdot\|$: 1) sample radius $r \sim \gammad{d+1}{1/\eps}$, the Gamma distribution with shape $d+1$ and scale $1/\eps$; 2) uniformly sample $z \sim B^d$; and 3) output $T(X) + rz$.
\end{lemma}

$\gammad{d+1}{1/\eps}$ can be sampled in $O(d)$, so sampling the $K$-norm mechanism reduces to sampling the norm unit ball $B^d$. Constructing these samplers is one of the main technical contributions of this work. Given statistic $T$, we choose a norm based on its \emph{sensitivity space}.

\begin{definition}[\citet{KN16, AS21}]
    The \emph{sensitivity space} of \\ statistic $T$ is $S(T) = \{T(X) - T(X') \mid X, X' \text{ are neighboring databases}\}$.
\end{definition}

By \Cref{lem:k_norm_dp}, given any norm with a unit ball that contains the convex hull of $S(T)$, the $K$-norm mechanism instantiated with that norm and $\Delta = 1$ is $\eps$-DP. We focus on cases where there is a norm whose unit ball is exactly the convex hull of $S(T)$.

\begin{lemma}
\label{lem:induced_norm}
    If set $W$ is convex, bounded, absorbing (for every $u \in \mathbb{R}^d$, there exists $c > 0$ such that $u \in cW$), and symmetric around 0 ($u \in W \Leftrightarrow -u \in W$), then the function $\|\cdot\|_W \colon \mathbb{R}^d \to \mathbb{R}_{\geq 0}$ given by $\|u\|_W = \inf\{c \in \mathbb{R}_{\geq 0} \mid u \in cW\}$ is a norm, and we say $W$ induces $\|\cdot\|_W$.
\end{lemma}

\citet{AS21} defined two orderings for comparing $K$-norm mechanisms and proved that induced norms are preferred in both orders.

\begin{lemma}[Theorem 3.19~\citet{AS21}]
\label{lem:optimal_k_norm}
        Let $\|\cdot\|_A$ and $\|\cdot\|_B$ be norms with associated unit balls $A$ and $B$. Let $M_V$ and $M_W$ be $K$-norm mechanisms instantiated with $\|\cdot\|_A$ and $\|\cdot\|_B$, respectively. Then we say $M_V$ is preferred over $M_W$ in \emph{containment order} if $\Delta_A \cdot A \subset \Delta_B \cdot B$, where $\Delta$ denotes sensitivity; we say $M_V$ is preferred over $M_W$ in \emph{volume order} if $|\Delta_A \cdot A| \leq |\Delta_B \cdot B|$, where $|\cdot|$ denotes Lebesgue measure.
        
        Suppose statistic $T$ has a sensitivity space $S(T)$ that induces norm $\|\cdot\|$, and let $M_V$ denote the corresponding $K$-norm mechanism. Then for any other norm $\|\cdot\|_K$ with associated $K$-norm mechanism $M_W$, $M_V$ is preferred over $M_W$ in both containment order and volume order.
\end{lemma}

\citet{AS21} further showed that better containment and volume orders also imply better entropy and conditional variance, among other notions. It follows that mechanisms which are optimal with respect to these orders are also optimal with respect to entropy and conditional variance (see Sections 3.2 and 3.3 of their paper for details). As our applications of these results are essentially immediate, we will not discuss them further. Nonetheless, they demonstrate that the three induced $K$-norm mechanisms we will construct enjoy unique utility guarantees.

The induced norm balls for the problems in this paper are all $d$-dimensional polytopes. The general state of the art for sampling these bodies is achieved by~\citet{LLV20}. They showed how to sample a $d$-dimensional polytope with $m$ constraints in time $\tilde O(md^{1+\omega})$, where $\omega \geq 2$ is the matrix multiplication exponent (Theorem 1.5 of~\citet{LLV20}). The polytopes considered in this paper have $\Omega(d)$ constraints, so this becomes $\tilde O(d^{2+\omega})$. Accounting for the mixing time to an approximation sufficient for $O(\eps)$-DP (Appendix A of~\citet{HT10}) increases the complexity to $O(d^{3+\omega})$. In contrast, the samplers introduced in this work are $\eps$-DP and have runtime $\tilde O(d^2)$.

Note that for consistency with the literature on sampling convex bodies, this paper defines time complexity as the number of field operations (addition and multiplication). In reality, runtime for these operations scales with input bit length; accounting for this increases complexity by roughly a factor of $d\log(d)$, as some of our algorithms involve arithmetic on $d$-bit numbers.

\subsection{Elliptic Gaussian Mechanism}
\label{subsec:prelims_ellipse}
Our second mechanism is elliptic Gaussian noise. It uses the fact that, to privately compute a statistic with sensitivity space $S$, it suffices to linearly transform the convex hull of $S$ to fit into the unit $\ell_2$ ball, add spherical Gaussian noise, and then invert the linear transformation as post-processing. Deriving these problem-specific linear transformations --- or, equivalently, computing minimum ellipses enclosing different sensitivity spaces --- is the other main technical contribution of this work.

\begin{lemma}[Adapted From \citet{NTZ13, NT23}]
\label{lem:elliptic_gaussian}
    Let $S$ be a convex body in $\mathbb{R}^d$ with $M \in \mathbb{R}^{d \times d}$ such that $S \subset MB_2^d$. Then the mechanism that on input $X \in S^n$ outputs $\sum_i X_i + Z$ where $Z \sim N(0, \tfrac{1}{2\rho}MM^T)$ is $\rho$-CDP.
\end{lemma}

The next lemma, proved in \Cref{sec:appendix_prelims}, establishes that sampling the $Z$ in \Cref{lem:elliptic_gaussian} reduces to sampling from a random scaling of $MB_2^d$, the ellipse containing the desired convex body. We therefore focus on deriving the ``best'' such ellipse, minimizing expected squared $\ell_2$ norm.

\begin{restatable}{lemma}{randomEllipse}
\label{def:random_ellipse}
    Let $E$ be an ellipse with axis lengths $\{a_1,...,a_d\}$ and corresponding orthonormal eigenvectors $\{v_1,...,v_d\}$. Let $D$ be the diagonal matrix where $D_{ii} = a_i$, and let $C$ be the matrix such that $Cv_i = e_i$ where $\{e_1,...,e_d\}$ is the standard basis. Let $M = C^{-1}DC$. Then $\countball \subset MB_{2}^{d}$, and drawing a uniform sample from $\mathcal{N}(0,MM^{T})$ reduces to uniform sampling from the random ellipse $RE$ where $R \sim \chi_{d}$, a Chi distribution with $d$ degrees of freedom.
\end{restatable}

The state of the art for finding these ellipses casts the problem as a semidefinite program (Theorem 32 of~\citet{NT23}). However, an approximately optimal solution is only guaranteed to be found in poly$(d)$ time for restricted classes of polytopes. Specifically, applying their result to our polytopes requires bounding the ``cotype-2 constant'' that arises from analyzing random walks in the dual polytope. We were not able to verify this bound for our polytopes, but even if we assume that it holds, the resulting algorithm relies on a sequence of oracles that all have unspecified poly$(d)$ runtimes. Unpacking the proofs of~\citet{NT23} and (generously) assuming linear runtimes for its constituent oracles yields a back of the envelope overall runtime of $O(d^5)$. In contrast, we explicitly identify closed-form expressions for exact minimum ellipses for our problems.

\subsection{Geometry}
For completeness, we briefly define vertices and other useful geometric terms.

\begin{definition}
\label{def:geometry}
    Let $X_n$ be any $n$-dimensional polyhedron in $\mathbb{R}^{d}$. For $1 \leq k \leq n-1$, we backwards inductively define $X_{k}$ to be all sets of the form $H_{k} \cap \partial{X_{k+1}}$ where $H_k$ is a $k$-dimensional (possibly affine) subspace in $\mathbb{R}^{d}$, $\partial{X_{k+1}}$ is the \emph{boundary} of $X_{k+1}$, and $\mu_k(H_{k} \cap \partial{X_{k+1}}) > 0$ where $\mu_k$ is $k$-dimensional Lebesgue measure. Lastly, we define $X_0$ to be the set $\partial{X_1}$. We call $X_k$ the $k$-dimensional \emph{faces} of $X_n$. Similarly, $X_0$ is the \emph{vertices} of $X_n$, and $X_1$ is the \emph{edges} of $X_n$. If two vertices are joined by an edge, we say that those vertices are \emph{neighboring}. For finite set $X$, let $CH(X)$ denote its \emph{convex hull}, and let $c(CH(X))$ be its \emph{center}, i.e., the mean of its vertices.
\end{definition}

Finally, we make a note about measure, often shorthanded ``volume'', that simplifies our sampling analysis by ignoring points with repeated coordinates. A proof appears in \Cref{sec:appendix_prelims}.

\begin{restatable}{lemma}{lebesgue}
\label{lem:lebesgue}
    Let $|U|$ denote the Lebesgue measure of set $U$, and let $E \subset[0,1]^{d}$ be the set of elements with repeated coordinates. Then $|E| = 0$.
\end{restatable}

\begin{assumption}
\label{lem:assumption}
    For the rest of this paper, whenever we consider a subset $X \subseteq [0,1]^{d}$ we will actually mean $X - E$, where $-$ denotes set difference, without explicitly writing this. By \Cref{lem:lebesgue}, this does not affect any of the subroutines that sample from a region of $[0,1]^{d}$ with nonzero measure.
\end{assumption}
\section{Sum}
\label{sec:sum}

\subsection{Sum Ball Sampler}
\label{subsec:sum}

Recall from the introduction that each Sum vector $x_i$ contains at most $k$ nonzero entries, each having absolute value at most $b$, and we compute the statistic $T = \sum_i x_i$. $b$ only affects scaling, so without loss of generality let $b=1$. We first derive the convex hull $\sumball$ of the sum sensitivity space

\begin{lemma}
\label{lem:sum_sensitivity_space}
    Let $B_{1,k}^d$ denote the $d$-dimensional $\ell_1$ ball of radius $k$ and let $B_\infty^d$ denote the $d$-dimensional $\ell_\infty$ unit ball. Then $\sumball =  B_{1,k}^d \cap B_\infty^d$, and $\sumball$ induces a norm.
\end{lemma}
\begin{proof}
    Since $T$ is a sum, $S(T) = \{T(X) - T(X') \mid X, X' \text{ are neighbors}\}$, the collection of all possible data vectors $X_i$ and their negations. Each point has $\leq k$ nonzero coordinates, each of which has absolute value  $\leq 1$, so the sensitivity space has vertices where between 1 and $k$ coordinates are $\pm 1$ and the remaining coordinates are 0. The convex hull of these vertices is $ B_{1,k}^d \cap B_\infty^d$.
    
    It remains to verify that $V$ induces a norm, using \Cref{lem:induced_norm}: $V$ is convex because it is a convex hull, bounded because it is an intersection of bounded sets, absorbing because it contains $B_{1,1}^d$, and symmetric around 0 because it is an intersection of symmetric sets.
\end{proof}

\begin{figure}
\centering
\tdplotsetmaincoords{75}{60}
\begin{tikzpicture} [scale=2.0, tdplot_main_coords, axis/.style={black,thick}, vector/.style={-stealth,black,very thick}, vector guide/.style={dashed,black,thick}]

        \coordinate (origin) at (0,0,0);
        \coordinate (k11) at (1,0,0);
        \coordinate (k12) at (0,1,0);
        \coordinate (k13) at (0,0,1);
        \coordinate (k21) at (0,1,1);
        \coordinate (k22) at (1,0,1);
        \coordinate (k23) at (1,1,0);
        \coordinate (k31) at (1,1,1);
        \coordinate (center) at (1,1,1);

        \draw[axis] (0,0,0) -- (1,0,0) node[anchor=north east]{};
        \draw[axis] (0,0,0) -- (0,1,0) node[anchor=north west]{};
        \draw[axis] (0,0,0) -- (0,0,1) node[anchor=south]{};

        \draw[fill=black] (origin) circle[radius=0.5pt] node[anchor=east]{origin};
        
        \foreach \aa in {k11,k12,k13}{
        \draw[line width=1pt] (origin) -- (\aa) node[yshift=0.5cm, anchor=south]{};
        }
        \foreach \aa in {k21,k22,k23}{
        \draw[line width=1pt] (k31) -- (\aa) node[yshift=0.5cm, anchor=south]{};
        }
        \foreach \aa in {k21,k22}{
        \draw[line width=1pt] (k13) -- (\aa) node[yshift=0.5cm, anchor=south]{};
        }
        \foreach \aa in {k21,k23}{
        \draw[line width=1pt] (k12) -- (\aa) node[yshift=0.5cm, anchor=south]{};
        }
        \foreach \aa in {k22,k23}{
        \draw[line width=1pt] (k11) -- (\aa) node[yshift=0.5cm, anchor=south]{};
        }
        
        \draw[line width=0.25mm,fill=blue!42,opacity=0.5] (k11) -- (k12) -- (k13) -- (k11);
        
        \draw[line width=0.5mm,fill=blue!42,opacity=0.0] (k21) -- (k22) -- (k23) -- (k21);
        
        \draw[line width=0.25mm,fill=blue!42,opacity=0.5] (k11) -- (k12) -- (k23) -- (k11);
        
        \draw[line width=0.25mm,fill=blue!42,opacity=0.5] (k13) -- (k21) -- (k22) -- (k13);
        
        \draw[line width=0.25mm,fill=blue!42,opacity=0.5] (k11) -- (k22) -- (k23) -- (k11);
        
        \draw[line width=0.25mm,fill=blue!42,opacity=0.5] (k13) -- (k22) -- (k11) -- (k13);
        
        \draw[line width=0.25mm,fill=blue!42,opacity=0.5] (k12) -- (k13) -- (k21) -- (k12);
        
        \draw[line width=0.25mm,fill=blue!42,opacity=0.5] (k12) -- (k23) -- (k21) -- (k12);
        
    \end{tikzpicture}
\tdplotsetmaincoords{75}{60}
\begin{tikzpicture} [scale=1.375, tdplot_main_coords, axis/.style={->,black,thick}, vector/.style={-stealth,black,very thick}, vector guide/.style={dashed,black,thick}]

        \coordinate (origin) at (0,0,0);
        \coordinate (k11) at (1,0,0);
        \coordinate (k12) at (0,1,0);
        \coordinate (k13) at (0,0,1);
        \coordinate (k21) at (0,1,1);
        \coordinate (k22) at (1,0,1);
        \coordinate (k23) at (1,1,0);
        \coordinate (k31) at (1,1,1);
        \coordinate (nk11) at (-1,0,0);
        \coordinate (nk12) at (0,-1,0);
        \coordinate (nk13) at (0,0,-1);
        \coordinate (nk21) at (0,-1,-1);
        \coordinate (nk22) at (-1,0,-1);
        \coordinate (nk23) at (-1,-1,0);
        \coordinate (nk31) at (-1,-1,-1);

        \draw[fill=black] (origin) circle[radius=0.5pt] node[anchor=north]{origin};

        
        \draw[line width=0.25mm,fill=blue!42,opacity=0.5] (k13) -- (k21) -- (k22) -- (k13);
        
        \draw[line width=0.25mm,fill=blue!42,opacity=0.5] (k11) -- (k22) -- (k23) -- (k11);
        
        \draw[line width=0.25mm,fill=blue!42,opacity=0.5] (k21) -- (k22) -- (k23) -- (k21);
        
        \draw[line width=0.25mm,fill=blue!42,opacity=0.5] (k12) -- (k21) -- (k23) -- (k12);
        
        \draw[line width=0.25mm,fill=blue!42,opacity=0.5] (nk13) -- (nk21) -- (nk22) -- (nk13);
        
        \draw[line width=0.25mm,fill=blue!42,opacity=0.5] (nk11) -- (nk22) -- (nk23) -- (nk11);
        
        \draw[line width=0.25mm,fill=blue!42,opacity=0.5] (nk21) -- (nk22) -- (nk23) -- (nk21);
        
        \draw[line width=0.25mm,fill=blue!42,opacity=0.5] (nk12) -- (nk21) -- (nk23) -- (nk12);
        
        \draw[line width=0.25mm,fill=blue!42,opacity=0.5] (k13) -- (k22) -- (nk12) -- (nk23) -- (k13);
    
        \draw[line width=0.25mm,fill=blue!42,opacity=0.5] (k13) -- (k21) -- (nk11) -- (nk23) -- (k13);
        
        \draw[line width=0.25mm,fill=blue!42,opacity=0.5] (k11) -- (k22) -- (nk12) -- (nk21) -- (k11);
        
        \draw[line width=0.25mm,fill=blue!42,opacity=0.5] (nk11) -- (k21) -- (k12) -- (nk22) -- (nk11);
        
        \draw[line width=0.25mm,fill=blue!42,opacity=0.5] (k11) -- (k23) -- (nk13) -- (nk21) -- (k11);
    
        \draw[line width=0.25mm,fill=blue!42,opacity=0.5] (k12) -- (k23) -- (nk13) -- (nk22) -- (k12);
        
    \end{tikzpicture}
\tdplotsetmaincoords{75}{60}
\begin{tikzpicture} [scale=0.8, tdplot_main_coords, axis/.style={->,black,thick}, vector/.style={-stealth,black,very thick}, vector guide/.style={dashed,black,thick}]

        \coordinate (origin) at (0,0,0);
        \coordinate (a) at (0,1,2);
        \coordinate (b) at (1,0,2);
        \coordinate (c) at (2,0,1);
        \coordinate (d) at (2,1,0);
        \coordinate (e) at (1,2,0);
        \coordinate (f) at (0,2,1);
        \coordinate (g) at (2,0,0);
        \coordinate (h) at (0,2,0);
        \coordinate (i) at (0,0,2);
        \coordinate (j) at (2,2,0);
        \coordinate (k) at (2,0,2);
        \coordinate (l) at (0,2,2);
        \coordinate (m) at (2,2,2);
        \coordinate (center) at (1,1,1);

        \draw[fill=black] (origin) circle[radius=0.5pt] node[anchor=north]{origin};
        
        \draw[line width=0.25mm,fill=blue!42,opacity=0.75] (a) -- (b) -- (c) -- (d) -- (e) -- (f) -- (a);
        
        \draw[line width=0.25mm,fill=blue!42,opacity=0.3] ($(a) - (m)$) -- ($(b) - (m)$) -- ($(c) - (m)$) -- ($(d) - (m)$) -- ($(e) - (m)$) -- ($(f) - (m)$) -- ($(a) - (m)$);
        
        \draw[line width=0.25mm,fill=blue!42,opacity=0.3] (a) -- (b) -- ($(b) - (m)$) -- ($(a) - (m)$) -- (a);
        
        \draw[line width=0.25mm,fill=blue!42,opacity=0.3] (b) -- (c) -- ($(c) - (m)$) -- ($(b) - (m)$) -- (b);
        
        \draw[line width=0.25mm,fill=blue!42,opacity=0.3] (c) -- (d) -- ($(d) - (m)$) -- ($(c) - (m)$) -- (c);
        
        \draw[line width=0.25mm,fill=blue!42,opacity=0.3] (d) -- (e) -- ($(e) - (m)$) -- ($(d) - (m)$) -- (d);
        
        \draw[line width=0.25mm,fill=blue!42,opacity=0.3] (e) -- (f) -- ($(f) - (m)$) -- ($(e) - (m)$) -- (e);
        
        \draw[line width=0.25mm,fill=blue!42,opacity=0.3] (f) -- (a) -- ($(a) - (m)$) -- ($(f) - (m)$) -- (f);
        
    \end{tikzpicture}
    \caption{Left: $R_{3,2}$ is the shaded region of the cube. Center: $\countball$, $k=2$; $R_{3,2}$ reappears in the upper right corner. Right: $\voteball$; $CH(P_3)$ is a regular polytope, but this is not true for general $d$.} \label{fig:illustrations}
\end{figure}
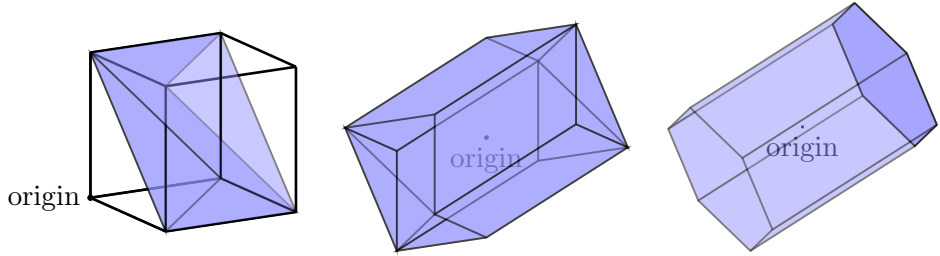

For both Sum and Vote (\Cref{sec:vote}), our sampler decomposes the polytope into simplices, randomly samples a simplex, and then returns a uniform sample from that simplex. We sample from the simplex using the following (folklore) result.

\begin{restatable}{lemma}{simplexSample}
\label{lem:simplex_sample}
    A collection of points $x_0, \ldots, x_{d} \in \mathbb{R}^n$ with $n \geq d$ are \emph{affinely independent} if $\sum_{i=0}^d \alpha_{i}x_i = 0$ and $\sum_{i=0}^d \alpha_i = 0$ implies $\alpha = 0$. A \emph{$d$-simplex} is the convex hull of $d+1$ affinely independent points and can be uniformly sampled in time $O(d\log(d))$.
\end{restatable}

The rest of this section is a simplified sketch of our sampler; a full exposition with pseudocode appears in \Cref{sec:appendix_sum}. The first step is to observe that, since $\sumball$ is symmetric around the origin, it suffices to uniformly sample the portion of $\sumball$ lying in the $\{+\}^d$ orthant (denoted $\sumball^+$) and then randomly permute its signs. Restricting attention to $\sumball^+$, we decompose it into $k$ ``slices''.

\begin{definition}
    For $j \in [k]$, define $H_j = \{x \in \mathbb{R}^d \mid \sum_{i=1}^d x_i \leq j\}$, $I_j = (0,1)^d \cap H_j$, and $R_j = I_j - I_{j-1}$ (sometimes denoted $R_{d,j}$ to make the ambient dimension $d$ explicit).
\end{definition}

Since $\cup_{j \in [k]} R_j = V^+$, the $R_j$ partition $\sumball^+$ (\Cref{fig:illustrations}). This decomposition is useful because it is closely connected to the sets of permutations with a fixed number of ascents.

\begin{definition}
\label{def:s_d}
    Let $S_d$ be the symmetric group on $d$ elements, i.e., the collection of permutations of $[d]$. Define the group action of $\sigma \in S_d$ on $x \in \mathbb{R}^{d}$ by $\sigma(x) = \sigma(x_1,...,x_d) = (x_{\sigma(1)},...,x_{\sigma(n)})$. For $X \subseteq \mathbb{R}^{d}$, define $\sigma(X) = \{\sigma(x) : x \in X\}$. A permutation $\sigma \in S_d$ has an \emph{ascent} at position $i$ if $\sigma(i) < \sigma(i+1)$. Let $S_{d, k} = \{\sigma \in S_d \mid \sigma \text{ has exactly } k \text{ ascents}\}$. For $d, j \in \mathbb{Z}_{\geq 0}$ the \emph{Eulerian number} $A_{d,j}$ is defined to be $|S_{d,j}|$.
\end{definition}

We can show that the cube $(0,1)^d$ may be partitioned into equal volume simplices, with exactly one simplex (of volume $1/(d!)$) for each permutation in $S_d$. Moreover, a similar bijection applies to individual slices, and each $R_j$ can be partitioned into $|S_{d,j-1}| = A_{d,j-1}$ simplices. It remains to (1) sample an $R_j$ from $\{R_j\}_{j=1}^k$ according to weights $\{A_{d,j-1}\}_{j=1}^{k}$, (2) uniformly sample a permutation from $S_{d,j-1}$, and (3) uniformly sample that permutation's corresponding simplex in $R_j$.

Step (1) uses the (folklore) identity $A_{x,y} = (x-y)A_{x-1,y-1} + (y+1)A_{x-1,y}$. Repeated application yields the relevant $A$ values for the weights in time $O(d^2)$.

Step (2) reuses these $A$ values. Having sampled slice index $j^*+1$, we uniformly sample $S_{d,j^*}$ by flipping a sequence of $d$ coins weighted by the $A$ values: starting with the first flip, a permutation in $S_{d,j^*}$ arises either from inserting an ascent into a permutation in $S_{d-1,j^*-1}$ or inserting a non-ascent into a permutation in $S_{d-1,j^*}$. We therefore apply the identity from step (1) and flip a coin with
\begin{equation*}
    \P{}{\text{heads}} = \frac{(d-j^*)A_{d-1,j^*-1}}{(d-j^*)A_{d-1,j^*-1} + (j^*+1)A_{d-1,j^*}}
\end{equation*}
and recursively sample $S_{d-1,j^*-1}$ if we get heads and $S_{d-1,j^*}$ if we get tails. This process determines when $j^*$ ascents are inserted during our final iterative construction of the permutation, though some additional care is required to ensure uniformity.

Finally, step (3) bridges the gap between discrete permutations and points in continuous space. To do so, we apply \Cref{lem:simplex_sample} to uniformly sample the ``fundamental simplex'' consisting of all points in the cube $(0,1)^d$ with increasing coordinates. Permuting the sample coordinates by the permutation from step (2) produces a uniformly sampled point with $j^*$ ascents. Finally, we apply an explicit bijection, constructed by~\citet{S77}, from such points to the points of $R_{j^*+1}$

The overall sampling time for $\sumball$ is dominated by the $O(d^2)$ computation of the $A$ values. We note that any subsequent samples only take time $O(d)$ each.

\subsection{Rejection Sampling the Sum Ball Is Inefficient}
\label{subsec:sum_rejection}
All of our rejection sampling results use the following result about $\ell_p$ ball volume.

\begin{lemma}[\cite{W05}]
\label{lem:lp_volume}
    Let $V_p^d(r)$ denote the volume of the $d$-dimensional $\ell_p$ ball of radius $r$. For $p \in [1,\infty)$, $V_p^d(r) = \left[2r\Gamma\left(1+\tfrac{1}{p}\right)\right]^d / \Gamma\left(1+\tfrac{d}{p}\right)$, and $V_\infty^d(r) = (2r)^d$.
\end{lemma}

It is easy to derive, for each $p \in [1,\infty]$, the minimum-radius $\ell_p$ ball around $\sumball$. The key technical step for our result is the following lemma, which we prove by analyzing the first and second derivatives of the expression in \Cref{lem:lp_volume} with respect to $p$.

\begin{restatable}{lemma}{OneOrInfBall}
\label{lem:1_or_inf} 
    The minimum-volume $\ell_p$ ball enclosing $\sumball$ is either the $\ell_1$ ball or the $\ell_\infty$ ball.
\end{restatable}

The remainder of the argument applies previous work bounding the volume of $\sumball$ to show that it is exponentially smaller than the $\ell_1$ or $\ell_\infty$ ball volumes given by~\Cref{lem:lp_volume}.
\section{Count}
\label{sec:count}

\subsection{Count Ball Sampler}
\label{subsec:count}
Recall that Count is Sum with an additional nonnegativity constraint.
\begin{lemma}
\label{lem:count_sensitivity_space}
    Let $V_+ = \{x \mid 0 \leq x_1, \ldots, x_d \leq 1 \text{ and } \|x\|_1 \leq k\}$. Then the convex hull of the count sensitivity space is $\countball = CH(V_+ \cup -V_+)$, and it induces a norm.
\end{lemma}
\begin{proof}
    By the same reasoning from \Cref{lem:sum_sensitivity_space}, the sensitivity space has vertices where between 1 and $k$ coordinates are nonzero. However, the nonnegativity constraint additionally means that the nonzero coordinates all have the same sign. This produces $\countball = CH(V_+ \cup -V_+)$.
    
    The same logic from \Cref{lem:sum_sensitivity_space} shows that $\countball$ is convex, bounded, and absorbing. Finally, it is symmetric around 0 because it is the convex hull of vertices that are symmetric around 0.
\end{proof}

$\countball$ is still symmetric around the origin, but it does not have the same shape in every orthant. Instead, we will see that the $2^d$ orthants fall into classes determined by the number of positive coordinates.

\begin{definition}
\label{def:j_stuff}
    Let $J_{0}^{d} = (1,...,1)$ be the vector of $d$ 1s, and define orthant $O(J_0^d) = \{x \in \mathbb{R}^d \mid x_1, \ldots, x_d \geq 0\}$. Given $J \in \{-1,1\}^d$, we define orthant $O(J) = \{J * v: v \in O(J_{0}^{d})\}$ where $*$ is element-wise multiplication, and define $J_+, J_- \subseteq [d]$ as the sets of coordinates at which $J$ equals 1 and -1, respectively. Finally, we define $V_J$ to be the vertices of $\countball$ in $O(J)$.
\end{definition}

Proofs of the following lemma, and other results in this section, appear in \Cref{sec:appendix_count}.

\begin{restatable}{lemma}{characterizationVJ}
\label{lem: characterization of V_J}
    Given $J \in \{-1,1\}^{d}$, $V_J$ consists of the subset of $V_{J_0}$ with support contained in $J_{+}$ and the subset of $V_{-J_0}$ with support contained in $J_{-}$.
\end{restatable}

\Cref{lem: characterization of V_J} is our primary tool for reasoning about the shape of $T(J) = CH(V_J)$ in each orthant $J$. It enables us to view the shape as an interpolation between the convex hull of its vertices with support contained in $J_+$ and the convex hull of its vertices with support contained in $J_-$. These convex hulls are identical to Sum balls with dimension $|J_+|$ and $|J_-|$, respectively. We can therefore reuse the knowledge of the Sum ball developed in \Cref{sec:sum}. However, the resulting argument is technically different. Instead of decomposing the relevant shape into simplices and reducing to an essentially discrete problem, we directly evaluate the integral for $|T(J)|$ by analyzing the infinitesimal ``shells'' of the interpolation. For $|J_+| = j$, this produces the expression
\begin{equation}
\label{eq:beta}
    |T(J)| = j\left(\sum_{i=1}^{k}\frac{A_{j,i-1}}{j!}\right)\left(\sum_{i=1}^{k}\frac{A_{d-j,i-1}}{(d-j)!}\right)\int_{0}^{1}t^{j-1}(1-t)^{d-j}\partial t
\end{equation}
where $A$ denotes Eulerian numbers (\Cref{def:s_d}). Evaluating the integral yields the following.
\begin{restatable}{lemma}{volumeOrthant}
\label{lem: volume of orthant}
    Given $J \in \{-1,1\}^d$ with $|J_{+}| = j$, $|T(J)| = \left(\sum_{i=1}^{k}A_{j,i-1}\right)\left(\sum_{i=1}^{k}A_{d-j,i-1}\right)\frac{1}{d!}$.
\end{restatable}

\Cref{lem: volume of orthant} provides the weights to sample an orthant index $J \in \{-1,1\}^d$ of $\countball$. To sample the orthant subshape $T(J)$, \Cref{eq:beta} shows that we can sample a cross-section of $T(J)$ by sampling a $\betad{j}{d-j+1}$ distribution, which has density $f(t) \propto t^{j-1}(1-t)^{d-j}$.

After sampling a cross-section index $t$, the last task is sampling the cross-section. We do so by decomposing the cross-section into subshapes, each of which is identical to a lower-dimensional Sum ball\footnote{A possible exception is a subshape identical to a hypersimplex (see \Cref{sec:appendix_count} for details).}, and then applying the Sum sampler from \Cref{sec:sum} twice, once for each of the two convex hulls in our interpolation. As a result, the final runtime is dominated by the $O(d^2)$ runtime of the two $\sumball$ samples. 

\subsection{Rejection Sampling the Count Ball Is Inefficient}
\label{subsec:count_rejection}
$\countball$ is contained inside $\sumball$ but has the same minimum containing $\ell_p$ balls, so a negative result for rejection sampling $\countball$ follows from the negative result for rejection sampling $\sumball$.

\subsection{Count Ellipse}
\label{subsec:count_ellipse}
This section derives a closed form for the $\ell_2^2$-minimizing ellipse containing $\countball$. We combine this with \Cref{lem:elliptic_gaussian} to obtain better Gaussian noise for Count.

\begin{definition}
\label{def:ellipse}
    A \emph{minimum ellipse} $E$ of a shape $X$ is an ellipse enclosing $X$ with minimum expected squared $\ell_2$ norm on the $d$-dimensional space it encloses, denoted $\enc{E}$. Given positive definite $A \in \mathbb{R}^{d \times d}$, we define $E_A = \{x \mid x^{T}Ax = 1\}$, sometimes denoted $E$ if $A$ is clear from context. Given a basis of eigenvectors $s_1, \ldots, s_d$ and eigenvalues $\lambda_1, \ldots, \lambda_d$ of $A$, $E_A$ has \emph{axis directions} $s_1, \ldots, s_d$ and \emph{axis lengths} $a_1 = 1/\sqrt{\lambda_1}, \ldots, a_d = 1/\sqrt{\lambda_d}$.
\end{definition}

The first result allows us to restrict attention to origin-centered ellipses. The proof argues that any ellipse not centered at the origin can be transformed into an origin-centered one with a strictly smaller expected squared $\ell_2$ norm. Proofs for this and the following results appear in \Cref{sec:appendix_count}.

\begin{restatable}{lemma}{ellipseCentering}
\label{lem:centering_lemma}
    Any minimum ellipse of $\countball$ is origin-centered.
\end{restatable}

Next, we relate an (origin-centered) ellipse $E_A$'s axis lengths to the magnitude of a random sample from $\enc{E_A}$. This will be useful for identifying a minimum ellipse.

\begin{restatable}{lemma}{ellipseNorm}
\label{lem:ellipse_squared_l2}
    Let ellipse $E_A$ have axis lengths $a_1, \ldots, a_d$, and let $Z$ be a uniform sample from $\enc{E_A}$. Then $\E{}{\|Z\|_2^2} = \tfrac{1}{d+2}\left(\sum_{j=1}^d a_j^2\right)$.
\end{restatable}

We now prove that the minimum ellipse of $\countball$ is unique. The proof analyzes the ``average'' ellipse that arises from combining two distinct minimum ellipses of $\countball$ and applies the Courant-Fischer theorem to argue that this average ellipse has smaller axes while still containing $\countball$. By the preceding lemma, this contradicts the assumption that the initial ellipses were minimal.

\begin{restatable}{lemma}{ellipseUnique}
\label{lem:count_ellipse_unique}
    The minimum ellipse of $\countball$ is unique.
\end{restatable}

It remains to derive explicit properties of this minimum ellipse, starting with its axes. The proof observes that transposing any two coordinates of the minimum ellipse produces another origin-centered ellipse containing $\countball$. By its minimality (\Cref{lem:ellipse_squared_l2}) and uniqueness (\Cref{lem:count_ellipse_unique}), this is exactly the minimum ellipse. Further analysis of the symmetries of the ellipse yields the claim.

\begin{restatable}{lemma}{ellipseAxes}
\label{lem:axes_directions}
    The minimum ellipse $E$ of $\countball$ has an axis along the $(1, \ldots, 1)$ direction, and the remaining axis lengths are equal, $a_2 = a_3 = \cdots = a_d$.
\end{restatable}

The final lemma identifies contact points between the minimum ellipse and $\countball$. This result relies on $k \leq d/2$. Informally, its proof argues that the polytope cross-section radius around the $(1, 1, \ldots, 1)$ vector varies as a parabola that peaks at $\|x\|_1 = d/2$, while the ellipse cross-section radius simply decreases with distance from the origin. For $k \leq d/2$, a minimum ellipse that contains the whole polytope must contact the polytope at the cross-section at $\|x\|_1 = k$. The argument does not extend to $k > d/2$ because the polytope cross-section radius is decreasing over this range.

\begin{restatable}{lemma}{ellipseContacts}
\label{lem:ellipse_contacts}
    For $k \leq d/2$, the minimum ellipse of $\countball$ contacts points with $k$ 1s and $d-k$ 0s.
\end{restatable}

This gives us constraints for a program to compute the minimum ellipse by minimizing the ellipse's squared axis lengths (\Cref{lem:centering_lemma} and \Cref{lem:ellipse_squared_l2}). Deriving a closed form solution via Lagrange multipliers yields \Cref{thm:count_ellipse}.

\begin{restatable}{theorem}{ellipseThm}
\label{thm:count_ellipse}
    For $k \leq d/2$, the minimum ellipse of $\countball$ can be computed in time $O(1)$.
\end{restatable}

A short note on parallelized generation of elliptic Gaussian noise appears in \Cref{sec:appendix_parallel}.
\section{Vote}
\label{sec:vote}

Recall that each vector $x_i$ is a permutation of $(0,1,\ldots,d-1)$, and we compute the statistic $T = \sum_i x_i$. The resulting sensitivity space is defined in part by permutohedra (\Cref{fig:illustrations}).

\begin{definition}
\label{def:permutohedron}
    Let $CH$ denote the convex hull, and let $P_d$ be the collection of all $d!$ permutations of $\{0,1,\ldots,d-1\}$. Then the \emph{permutohedron} is $CH(P_d)$.
\end{definition}

\begin{lemma}
\label{lem:vote_sensitivity_space}
    The convex hull of the sensitivity space associated with vote is $\voteball = CH(P_d \cup -P_d)$, and $\voteball$ induces a norm.
\end{lemma}
\begin{proof}
    Since $T$ is a sum, $S(T) = \{T(X) - T(X') \mid X, X' \text{ are neighbors}\}$ consists of all possible points and their negations. Thus, any point in $S(T)$ either has all nonnegative coordinates or all nonpositive coordinates, and the vertices of $S(T)$ are $P_d \cup -P_d$.
    
    Recalling \Cref{lem:induced_norm}, $\voteball$ is convex and bounded because it is the convex hull of a finite set. For any point $x \in P_d$, every point on the line between $x$ and $-x \in -P_d$ is also in $\voteball$; 0 is on the line between $(1,1,\ldots,1) \in CH(P_d)$ and $(-1,-1,\ldots,-1) \in CH(-P_d)$, so this implies the existence of a neighborhood around 0 in $V$, and $\voteball$ is absorbing. Finally, any $x = (x_1, \ldots, x_d) \in \voteball$ lies in some translation $Y$ of $P_d$ (along the $I_{[d]}$ axis) between $P_d$ and $-P_d$. 
    Let $f:Y\rightarrow Y$ be the map that reflects a point in $Y$ across $c(Y)$. Let $g:Y\rightarrow -Y$ be the map that reflects a point in $Y$ across the hyperplane $x_1 + ... + x_d = 0$. Then the action of $g \circ f$ is to move a point $x \in Y$ to the point diagonal from it on the rectangle with vertices $x, f(x), g(f(x)), g(x)$. On the other hand, the center of the rectangle is at the origin, so the action of $g \circ f$ is equal to the action of the map $x \rightarrow -x$. As $\mbox{Image}(g\circ f) = -Y$, then $-x \in \voteball$.
\end{proof}

\subsection{Vote Sampler}
\label{subsec:vote}

Our goal is to sample from $\voteball$, a cylinder whose bases are positive and negative permutohedra. We start by expressing the $(d-1)$-dimensional positive permutahedron as a ``star decomposition'' into $(d-1)$-dimensional pyramids, each of which have the center of the permutahedron as a common apex and a $(d-2)$-dimensional face of the permutohedron as a base. To sample a pyramid, we need to know the types of pyramids and their volumes. The following lemma is a first step to both. It is a simplified version of a statement given (without proof) by~\citet{P09}; a proof of the full statement appears in \Cref{sec:appendix_vote}, along with proofs of other results and pseudocode.

\begin{lemma}
\label{lem:simple_face_decomposition}
    There is a bijection between the $(d-2)$-dimensional faces of $CH(P_d)$ and the ordered pairs of subsets partitioning $[d]$. Moreover, let $F$ be a $(d-2)$-dimensional face of $CH(P_d)$ corresponding to subsets $B_1, B_2$, and for $i = 1, 2$, let $I_{B_i}$ be the vector with 1s at the indices in $B_i$ and 0s elsewhere. Then $F = (CH(P_{B_1}) + (d-|B_1|)I_{B_1}) \oplus (CH(P_{B_2})$, where for $J \subset [d]$, $P_J$ is an embedding of $P_{|J|}$ at the coordinates of of $J$.
\end{lemma}

Note that the face is a direct sum of subpermutohedra. This will eventually yield a recursive algorithm that samples from successively smaller subpermutohedra.

Next, we compute the counts and volumes of these faces. The counts follow from \Cref{lem:simple_face_decomposition}. The proof of the volumes relies on existing results for permutohedra volume~\citep{ASV21, S86}, though some additional work is required to derive an explicit formula.

\begin{restatable}{lemma}{volumeCalculation}
\label{lem:volume calculation}
    Let $F$ be a $(d-2)$-dimensional face of $CH(P_d)$ corresponding to $B_1, B_2$. There are $\binom{d}{|B_1|}$ faces congruent to $F$ and each has $(d-2)$-volume $|B_1|^{|B_1|-3/2}|B_2|^{|B_2|-3/2}$.
\end{restatable}

Having analyzed the pyramid bases, we now turn to the pyramid heights. This mostly follows from the subpermutohedron decomposition given in \Cref{lem:simple_face_decomposition}.

\begin{restatable}{lemma}{altitudeCalculation}
\label{lem:altitude calculation}
    Let $F$ be a $(d-2)$-dimensional face of $CH(P_d)$ corresponding to $B_1, B_2$. Then the vector from $c(CH(P_d))$ to $c(F)$, where $c(\cdot)$ denotes center, is orthogonal to $F$ and has length $\frac{1}{2}\sqrt{|B_1||B_2|^{2} + |B_2||B_1|^{2}}$.
\end{restatable}

This enables us to sample one of the $(d-1)$-dimensional pyramids composing $CH(P_d)$. It remains to sample a point from the chosen pyramid. We again rely on decomposition into simplices. We use \Cref{lem:simple_face_decomposition} to prove that it suffices to recursively sample a simplex from a star decomposition of each of these subpermutohedra. 

\begin{restatable}{lemma}{simplexProductTriangulation}
\label{lem:simplex_product_triangulation}
    Let $\Delta_x$ be an $n$-simplex in $\mathbb{R}^{n+m}$ with vertices $\{x_0,...,x_n\}$ where $x_0 = 0$ and $\Delta_x$ lives in the subspace $V_x$ of the first $n$ coordinates. Let $\Delta_y$ be an $m$-simplex in $\mathbb{R}^{n+m}$ with vertices $\{y_0,...,y_m\}$ where $y_0 = 0$ and $\Delta_y$ lives in the subspace $V_y$ of the last $m$ coordinates. Let $D$ be the set of $(n+m)$-simplices formed by any sequence starting with $x_0 \oplus y_0$, ending with $x_n \oplus y_m$, and with the property that $x_i \oplus y_j$ is followed by either $x_{i+1} \oplus y_j$ or $x_{i} \oplus y_{j+1}$. Then $D$ decomposes $\Delta_x \oplus \Delta_y$ into equal volume simplices.
\end{restatable}

After sampling a $(d-2)$-dimensional simplex $\Delta_{d-2}$ uniformly from the base $F$ of a pyramid, we can form the $(d-1)$-dimensional simplex $\Delta_{d-1}$ by connecting the vertices of $\Delta_{d-2}$ to $c(CH(P_d))$. Then $\Delta_{d-1}$ is a simplex sampled with the appropriate probability from a simplex decomposition of $CH(P_d)$. We apply \Cref{lem:simplex_sample} to uniformly sample $z$ from $\Delta_{d-1}$. Finally, sampling from the cylinder $\voteball$ is easy: uniformly sample from the line between $z \in CH(P_d)$ and its reflection $z' = z - (d-1)I_{[d]}$ in $-CH(P_d)$.

The overall $O(d^2\log(d))$  runtime for sampling $\voteball$ given in \Cref{thm:informal} comes from the $O(d)$ subpermutohedra recursions and the $O(d\log(d))$ time spent computing pyramid sampling weights in each recursion.

\subsection{Rejection Sampling the Vote Ball Is Inefficient}
\label{subsec:vote_rejection}
As in \Cref{subsec:sum_rejection}, we derive the radius of the minimium $\ell_p$ ball enclosing $\voteball$.

\begin{restatable}{lemma}{VoteEnclosing}
\label{lem:vote_enclosing}
    For $p \in [1,\infty)$, the minimum $r(p)$ such that $r(p)B_p^d$ contains $\voteball$ is $r(p) = \left(\sum_{j=0}^{d-1} j^p\right)^{1/p}$, and $r(\infty)=d-1$.
\end{restatable}

With this result, showing that rejection sampling $\voteball$ using an $\ell_p$ ball is inefficient again reduces to lower bounding the volumes of the enclosing $\ell_p$ balls.

\begin{restatable}{theorem}{voteRejectionSample}
\label{thm:vote_rejection_sample}
    For any $p \in [1,\infty]$, rejection sampling $\voteball$ using the minimum enclosing $\ell_p$ ball takes at least $\frac{(1.77)^{d}}{4}$ samples in expectation for $d \leq p$, and $\frac{(1.2)^{d-1}}{d}$ samples for $d > p$.
\end{restatable}

\subsection{Vote Ellipsoid}
\label{subsec:ellipsoid}
We now turn to a closed form for the $\ell_2^2$-minimizing ellipse containing $\voteball$. The first lemma proceeds from the same arguments used to prove \Cref{lem:centering_lemma} and \Cref{lem:count_ellipse_unique}, as $\voteball$ is also origin-centered and symmetric around the origin.

\begin{lemma}
\label{lem:vote_ellipse_properties}
    Any minimum ellipse of $\voteball$ is origin-centered and unique.
\end{lemma}

Its axis directions are also identical to those of $\countball$. The proof from \Cref{lem:axes_directions} still applies, because transposing arbitrary coordinates of any vertex in $\voteball$ produces another vertex in $\voteball$; see \Cref{lem:transposition} in the Appendix for details.

\begin{lemma}
\label{lem:vote_axes_directions}
    The minimum ellipse of $\voteball$ has an axis along the $(1, \ldots, 1)$ direction, and the remaining axis lengths are equal, $a_2 = a_3 = \cdots = a_d$.
\end{lemma}

It remains to find a contact point between the minimum ellipse and $\voteball$. The minimum ellipse must contact at least one vertex of $\voteball$, but because of \Cref{lem:vote_ellipse_properties} and \Cref{lem:vote_axes_directions}, and the fact that all elements of $CH(P_d)$ are equidistant from the $(1, 1, \ldots, 1)$ axis, contacting one means that it contacts all of them.

\begin{restatable}{lemma}{voteEllipseContacts}
\label{lem:vote_ellipse_contacts}
    The minimum ellipse of $\voteball$ contacts the vertices of $CH(P_d)$.
\end{restatable}

This again yields a program that can be solved using Lagrange multipliers.

\begin{restatable}{theorem}{voteEllipseTheorem}
    The minimum ellipse of $\voteball$ can be computed in time $O(1)$.
\end{restatable}
\section{Acknowledgments}
We thank Jennifer Gillenwater, Alex Kulesza, Nikita Lvov, Andr\'es Mu\~noz Medina, and Thomas Steinke for helpful comments.

\newpage

\bibliographystyle{plainnat}
\bibliography{references}

\newpage

\section{Proofs For Preliminaries}
\label{sec:appendix_prelims}

\randomEllipse*
\begin{proof}
    Note that $Mv_i = C^{-1}DCv_i = C^{-1}De_i = C^{-1}(a_{i}e_{i}) = a_{i}v_{i}$, so $M$ is the linear transformation that scales eigenvector $v_i$ by $a_i$. In other words, $MB_{2}^{d} = E$, so $\countball \subset MB_{2}^{d}$. Since for $i \in [d]$ we have $C^{-1}e_i = v_i$, the columns of $C^{-1}$ are $\{v_1, \ldots, v_d\}$. Similarly, $Cv_i = e_i$ implies that the rows of $C$ are $\{v_1, \ldots, v_d\}$, so $C$ is unitary, and $MM^{T} = C^{-1}D^{2}C = (C^{T}D)(DC) = (C^{T}D)(C^{T}D)^{T}$. It follows that $\mathcal{N}(0,MM^{T}) = C^{T}D\mathcal{N}(0,I_d)$. 
    
    Suppose $X \sim \mathcal{N}(0,I_d)$. Equivalently, $X$ is generated by first drawing a radius $R$ from a Chi distribution $\chi_d$, sampling $Y$ from the unit sphere, and computing $X = RY$. As $RY$ is a uniform sample from $RB_{2}^{d}$, $C^{T}DX = C^{T}DRY$ is a uniformly random sample from $RE$ (since the linearity of the transform preserves uniformity).
\end{proof}

\lebesgue*
\begin{proof}
    Each $x \in E$ induces an equivalence class partition of indices $C = \{I_{1},...,I_{n}\}$ where $I_j \subset \{1,2,...,d\}$ and indices $i,j \in \{1,...,d\}$ are equivalent if $x_i = x_j$. Define $V_C = \mbox{span}\{v_1,...,v_n\}$ where $v_j \in \{0,1\}^{d}$ is the vector with coordinates equal to 1 exactly at each index in $I_j$. Since $n < d$, $|V_C| = 0$. As there are finitely many possible equivalence class partitions of indices, say $\{C_1,...,C_m\}$, then $E \subseteq \cup_{i=1}^{m}V_{C_i}$ and $0 \leq |E| \leq \sum_{i=1}^{m}|V_{C_i}| = 0$ so $|E| = 0$. 
\end{proof}
\section{Proofs For Sum}
\label{sec:appendix_sum}

\subsection{Proofs For Sum Sampler}

\simplexSample*
\begin{proof}
    Denote the simplex in question by $\Delta$, with vertices $x_0, \ldots, x_d$. By definition, each point of $\Delta$ can be expressed as a convex combination of $x_0, \ldots, x_d$. If we have two such convex combinations $\sum_{i=0}^d \alpha_{i}x_i$ and $\sum_{i=0}^d \beta_{i}x_{i}$ with distinct $\alpha$ and $\beta$, then $\sum_{i=0}^d (\alpha_i - \beta_i)x_i = 0$, and $\sum_{i=0}^d (\alpha_i - \beta_i) = 1-1 = 0$, so affine independence implies $\alpha = \beta$. It follows that every point from $\Delta$ has a unique expression as a convex combination of $x_0, \ldots, x_d$. 
    
    Let $B = \{e_1,..,e_d\}$ be the standard basis in $\mathbb{R}^{d}$. We will show that a uniform distribution over the basis $B$ corresponds to a uniform distribution when we change to the basis $B_{x} = \{x_1,...,x_d\}$. Let $f$ be the uniform density function over the simplex with vertices in $B_{x}$. Then $\int_{x \in \Delta}f dB = 1$. Let $M$ be the matrix whose $i$th row is equal to $x_i$ written with coordinates in $B$. When we switch from integration over $B$ to integration over $B_{x}$, we need to calculate the Jacobian matrix which is $M^{-1}$.  Then $1 = \int_{x \in \Delta}f dB = \int_{x \in \Delta_s}f |\det{M^{-1}}|dB_{x}$ where $\Delta_s$ is the standard simplex, i.e., the simplex with vertices in $B$. Since $f$ is uniform, it follows that $f |\det{M^{-1}}|$ is a uniform density function over $\Delta_s$ when we switch to the $B_{x}$ basis, so sampling a point uniformly from $\Delta$ in the $B$ basis corresponds to sampling a point uniformly from $\Delta_s$ in the $B_{x}$ basis. We can do the latter in time $O(d\log(d))$ by drawing $d-1$ samples from $U(0,1)$, appending 0 and 1, sorting the $d+1$ elements, and taking the $d$ distances $\{\alpha_0,...,\alpha_d\}$ between adjacent elements~\cite{R81}. Then we return $\sum_{i=0}^{d}\alpha_{i}x_i$.
\end{proof}

We start by defining the fundamental simplex.

\begin{definition}
\label{def:fundamental_simplex}
    An \emph{open} simplex is a simplex minus its boundary. The \emph{fundamental $d$-simplex} $\Delta_d$ is the open simplex with vertices $\{f_0, f_1,f_2,...,f_d\}$ where $f_i \in \{0,1\}^{d}$ is the vector whose first $d-i$ coordinates are 0 and whose last $i$ coordinates are 1.
\end{definition}

We will repeatedly view points in $(0,1)^d$ as permutations of points in $\Delta_d$. \Cref{def:fundamental_simplex} makes it clear that $\Delta_d$ is a simplex, but the following lemma provides an equivalent description that will be easier to reason about algebraically.

\begin{lemma}
\label{lem:fundamental_simplex_increasing}
    The fundamental simplex $\Delta_d = \{x \in (0,1)^d: x_1 < ... < x_d\}$.
\end{lemma}
\begin{proof}
    Given $x \in \Delta_d$, it is a convex combination of $\{f_0, f_1, \ldots, f_d\}$, so we can write $x = \sum_{i=0}^d c_{i}f_{i}$  where $c_i \in (0,1)$ and $\sum_{i=0}^d c_{i} = 1$. Then $x_j = \sum_{i=d-j+1}^d c_i$ for all $1\leq j \leq d$, so $x_1 < ... < x_d$. Conversely, given $x \in (0,1)^d$ with $x_1 < ... < x_d$, then we can define $c_d = x_1$, for $2 \leq j \leq d$ define $c_{d-j+1} = x_j - x_{j-1}$, and finally define $c_0 = 1 - x_d$ so $\sum_{i=0}^d c_i = 1$. Then $(c_df_d)_1 = c_d = x_1$, $(c_df_d + c_{d-1}f_{d-1})_2 = c_d + c_{d-1} = x_2$, and in general $x = \sum_{i=0}^d c_{i}f_{i}$ is a convex combination of $\{f_0, \ldots, f_d\}$.
\end{proof}

To connect regions and permutations, we apply $S_d$ to $\Delta_d$ to obtain a partition of $(0,1)^d$.

\begin{lemma}
\label{lem:partition_cube}
    $S_{d}(\Delta_d) = \{\sigma(\Delta_d) : \sigma \in S_d\}$ partitions $(0,1)^{d}$ into disjoint open simplices.
\end{lemma}
\begin{proof}
    For $\sigma \in S_d$, $\sigma(\Delta_d) = \{(x_{\sigma(1)},...,x_{\sigma(d)}) : x \in \Delta_d\} = \{x \in (0,1)^{d} : x_{\sigma^{-1}(1)} < ... < x_{\sigma^{-1}(d)}\}$. For every $x \in (0,1)^{d}$ there is exactly one $\sigma_x \in S_d$ such that $x_{\sigma_x^{-1}(1)} < ... < x_{\sigma_x^{-1}(d)}$, so $x \in \sigma_x(\Delta_d)$.
\end{proof}

Moreover, there is a concrete bijection between regions $\sigma(\Delta_d)$ and permutations.

\begin{lemma}
\label{lem:G_dk_map}
    Fix $0 \leq k < d$. Let $T_{d,k} = \{\sigma(\Delta_d) \in S_d(\Delta_d): \mbox{every } x \in \sigma(\Delta_d) \mbox{ has exactly k ascents} \}$. Then $T_{d,k} = \{\sigma(\Delta_d): \sigma \in S_{d,k}\}$ and, defining $G_d(\sigma) = \sigma(\Delta_d)$, its restriction $G_{d,k}$ to $S_{d,k}$ is a bijection between $S_{d,k}$ and $T_{d,k}$.
\end{lemma}
\begin{proof}
    $x \in \sigma(\Delta_d) \in T_{d,k}$ if and only if $x$ has exactly $k$ ascents and $x = (x_{\sigma(1)}', \ldots, x_{\sigma(d)}')$  for some $x' \in \Delta_d$. $x' \in \Delta_d$ if and only if $x_1' < \cdots < x_d'$. Thus $x_{\sigma(i)}' < x_{\sigma(i+1)}'$ if and only if $\sigma(i) < \sigma(i+1)$. Thus, $x$ has exactly $k$ ascents if and only if $\sigma$ has exactly $k$ ascents, so $T_{d,k} = \{\sigma(\Delta_d) \mid \sigma \in S_{d,k}\}$. To see that $G_{d,k}$ is a bijection, we use $G_{d,k}^{-1}(\sigma(\Delta_d)) = \sigma$.
\end{proof}

Recapping the argument so far, the slices $R_1, \ldots, R_k$ partition $V^+$, permuting $\Delta_d$ partitions $(0,1)^d$ into simplices (\Cref{lem:partition_cube}), and there is a bijection between those simplices and partitions in terms of ascents (\Cref{lem:G_dk_map}). The last step connecting regions and permutations relies on an explicit map $\varphi$ introduced by~\citet{S77} .

\begin{lemma}[\cite{S77}]
\label{lem:phi_bijection}
    Define $\varphi \colon (0,1)^d \to (0,1)^d$ by $\varphi(x) = y$ where $y_j = x_{j-1} - x_j + \indic{x_{j-1} < x_j}$ and  we define $x_0 = 0$. Except on a set of measure 0, $\varphi$ is a measure-preserving bijection from $U_j = \{x \in (0,1)^d \mid x \text{ has exactly } j \text{ ascents}\}$ to $R_{j+1}$. 
\end{lemma}

The following lemma brings these ideas together by using $\varphi$ to compute the volumes of the $R_j$ slices. Perhaps unsurprisingly, the volumes are characterized by counting permutations.

\begin{lemma}
\label{lem:R_volume}
    For $d, j \in \mathbb{Z}_{\geq 0}$ define \emph{Eulerian number} $A_{d,j} = |\{\sigma \in S_d \mid \sigma \text{ has exactly } j \text{ ascents}\}|$. Then the $d \times d$ table $A$ can be computed in time $O(d^2)$. Moreover, for $j \in [k]$, $|R_j| = A_{d,j-1} / (d!)$.
\end{lemma}
\begin{proof}
    To compute $A$, we repeatedly apply the (folklore) identities $A_{x,y} = (x-y)A_{x-1,y-1} + (y+1)A_{x-1,y}$ and $A_{0,0} = 1$ and $A_{0,y} = 0$ for all $y \neq 0$.
    
    $|R_j| = A_{d,j-1} / (d!)$ has been described as ``implicit in the work of Laplace''~\citep{S77}, but we prove it explicitly here. First, we can rewrite $\varphi(x) = Mx + b$, where $M$ is lower triangular with -1's on the diagonal, 1's on the subdiagonal, and 0's elsewhere, and $b_j$ is the indicator that $x_{j-1} < x_j$. Note that, for any fixed $\sigma$, $b$ is constant over $x \in \sigma(\Delta_d)$. As $\sigma(\Delta_d)$ is convex with vertices $\{\sigma(f_1),...,\sigma(f_d)\}$, $M(\sigma(\Delta_d))$ is convex with vertices $\{M(\sigma(f_1),...,M(\sigma(f_d))\}$, i.e. $M(\sigma(\Delta_d))$ is a simplex and so is its translation $M(\sigma(\Delta_d)) + b$. Then $\det(M) = (-1)^{d}$, so as a volume-preserving transformation of the fundamental $d$-simplex, which has volume $\tfrac{1}{d!}|\det(f_1 - f_0, \ldots, f_d - f_0)| = \tfrac{1}{d!}$, we get $|\sigma(\Delta_d)| = |M(\sigma(\Delta_d))| = |\varphi(\sigma(\Delta_d))| = 1/(d!)$.
    
    By \Cref{lem:G_dk_map}, $G_{d,j-1}(S_{d,j-1}) = \{\sigma(\Delta_d) \mid \sigma \in S_{d,j-1}\} = T_{d,j-1}$ partitions $U_{d,j-1}$ into simplices. Thus $\{\varphi(\sigma(\Delta_d)) : \sigma \in S_{d,j-1}\}$ partitions $R_{d,j}$ into $A_{d,j-1}$ simplices, and $|R_{d,j}| = A_{d,j-1}/(d!)$. 
\end{proof}

We have established how to sample a slice $R_{d,j}$ proportionally to its volume. The remaining task is to sample uniformly from $R_{d,j}$. By \Cref{lem:phi_bijection} and \Cref{lem:R_volume}, $R_{d,j}$ admits a partition into $A_{d,j-1}$ simplices, each of which corresponds to a unique $\sigma \in S_{d,j-1}$. Thus, two steps remain: uniformly sampling a permutation from $S_{d,j-1}$, and finally uniformly sampling a point from the associated simplex (\Cref{lem:simplex_sample}).

\begin{lemma}
\label{lem:sampling S_{d,k}}
    We can uniformly sample an element of $S_{d,j}$ in time $O(d^2)$.
\end{lemma}
\begin{proof}
    Viewing permutation $\sigma$ as the list $\{\sigma(1),...,\sigma(d)\}$, any $\sigma \in S_d$ with $j$ ascents arises from two possibilities of inserting $d$ into a permutation  $\sigma \in S_{d-1}$. There are $d$ possibilities for insertion (at the beginning of the list, between two elements, and at the end of the list), so the two possible cases are
    \begin{enumerate}
        \item $\sigma \in S_{d-1,j-1}$. Then inserting $d$ increases the number of ascents, so $d$ must be inserted in a place in $\sigma$ that is currently a descent or at the end of the list. $\sigma$ has $j-1$ ascents, and of the remaining $d-(j-1) = d+1-j$ spots, one is at the beginning of the list, where inserting $d$ would not increase the number of ascents. Thus, there are $d-j$ possible places.
        \item $\sigma \in S_{d-1,j}$. Then inserting $d$ maintains the number of ascents, so $d$ must be inserted in a place in $\sigma$ that is currently an ascent, or at the beginning. $\sigma$ has $j$ ascents, so there are $j+1$ possible places.
    \end{enumerate}
    Thus to sample a uniformly random element of $S_{d,k}$, we first flip a coin with probability of heads \begin{equation*}
        \frac{(d-k)A_{d-1,k-1}}{(d-k)A_{d-1,k-1} + (k+1)A_{d-1,k}}.
    \end{equation*}
    If heads, we recursively uniformly sample a random element of $S_{d-1,k-1}$. If tails, we recursively uniformly sample a random element of $S_{d-1,k}$. At the end of the process, we have a sequence of $d$ coin flips with $j$ heads and $d-j$ tails. Starting from the permutation $(1)$, we successively add $2, 3, \ldots, d$ by either inserting it in one of the current descents or end of the list (if heads) or the current ascents or beginning of the list (if tails), choosing the position uniformly at random.

    By $A_{x,y} = (x-y)A_{x-1,y-1} + (y+1)A_{x-1,y}$, flipping the $d$ coins and building the permutation each take $O(d^2)$ arithmetic operations.
\end{proof}

Having described the sampler components, we collect them into \Cref{alg:sum_sampler}, and the final guarantee is \Cref{thm:sample_sum}.

\begin{theorem}
\label{thm:sample_sum}
    The polytope $V$ described in \Cref{lem:sum_sensitivity_space} can be sampled in time $O(d^2)$.
\end{theorem}

\begin{algorithm}
    \begin{algorithmic}[1]
    \STATE {\bfseries Input:} Dimension $d$ and $\ell_0$ bound $k$
    \FOR{$j=1, \ldots, k$}
        \STATE Compute $|R_j|$ using \Cref{lem:R_volume}
    \ENDFOR
    \STATE Sample $j \propto |R_j|$
    \label{alglnlabel:sum_slice_sample}
    \STATE Uniformly sample $\sigma$ from $S_{d,j-1}$ using \Cref{lem:sampling S_{d,k}}
    \STATE Sample $x$ from fundamental simplex $\Delta_d$ using \Cref{lem:simplex_sample}
    \STATE Compute $y = \varphi(\sigma(x))$ using \Cref{lem:phi_bijection}
    \STATE Randomly set the sign of each coordinate of $y$
    \STATE Return $y$
    \caption{Sum Sampler}
    \label{alg:sum_sampler}
    \end{algorithmic}
\end{algorithm}

\subsection{Proofs For Sum Rejection Sampling}

We first derive the radius of the minimum $\ell_p$ ball enclosing $V = kB_1^d \cap B_\infty^d$ (\Cref{lem:sum_sensitivity_space}).

\begin{lemma}
\label{lem:sum_enclosing}
    For $p \in [1,\infty)$, the minimum $r$ such that $rB_p^d$ contains $V$ is $r = k^{1/p}$. The minimum $r$ such that $rB_\infty^d$ contains $V$ is $r=1$.
\end{lemma}
\begin{proof}
    The $\ell_p$ norm of points from $kB_1^d$ is maximized at the vertices on the axes, so the maximum $\ell_p$ norm of a point in $V$ is at any of the vertices consisting of $k$ coordinates of $\pm 1$ and $d-k$ coordinates of 0, which have norm $k^{1/p}$ for $p < \infty$ and 1 for $p=\infty$.
\end{proof}

The next step shows that it suffices to restrict our attention to the two extremes $p=1$ and $p=\infty$. The analysis reduces to two cases: when $k$ is large, the $\ell_p$ ball volume is minimized at $p=\infty$, and when $k$ is small, it is minimized at either $p=1$ or $p=\infty$.

\OneOrInfBall*
\begin{proof}
    Since $\ell_p$ balls are symmetric across orthants, we drop the $2^d$ factor in \Cref{lem:lp_volume} and focus on single-orthant volume.
    By \Cref{lem:lp_volume} and \Cref{lem:sum_enclosing}, the one-orthant volume of the minimum $\ell_p$ ball enclosing $V$ is
    \begin{equation}
    \label{eq:ball_vol}
        W_{p}^{d}(k^{1/p}) = \frac{\left[k^{1/p}\Gamma\left(1+\tfrac{1}{p}\right)\right]^d}{ \Gamma\left(1+\tfrac{d}{p}\right)}
    \end{equation}
    We will use the following result to analyze how $\ell_p$ ball volume changes with $p$.
    \begin{claim}{4.3.1}
    \label{claim:gamma_derivative}
        $\frac{\partial}{\partial p} \frac{\Gamma(1+\frac{1}{p})^d}{\Gamma(1+\frac{d}{p})} = \frac{d \cdot \Gamma(1+\frac{1}{p})^d}{p^2\Gamma(1+\frac{d}{p})} \cdot \left[\psi\left(\frac{d}{p}\right) + \frac{p}{d} - \psi\left(\frac{1}{p}\right) - p\right]$.
    \end{claim}
    \begin{proof}
    $\Gamma'(x) = \Gamma(x)\psi(x)$ where $\psi$ is the digamma function, so
    \begin{align*}
        \frac{\partial}{\partial p} \frac{\Gamma(1+\frac{1}{p})^d}{\Gamma(1+\frac{d}{p})} =&\ \frac{\Gamma(1+\frac{d}{p}) \cdot d \cdot \Gamma(1+\frac{1}{p})^{d-1} \cdot \frac{\partial}{\partial p}\Gamma(1+\frac{1}{p}) - \Gamma(1+\frac{1}{p})^d \frac{\partial}{\partial p}\Gamma(1+\frac{d}{p})}{\Gamma(1+\frac{d}{p})^2} \\
        =&\ \frac{-\Gamma(1+\frac{d}{p}) \cdot d \cdot \Gamma(1+\frac{1}{p})^{d-1} \cdot \Gamma(1+\frac{1}{p})\psi(1+\frac{1}{p}) + \Gamma(1+\frac{1}{p})^d}{p^2\Gamma(1+\frac{d}{p})^2} \\
        +&\ \frac{d \cdot \Gamma(1+\frac{1}{p})^d \cdot \Gamma(1+\frac{d}{p})\psi(1+\frac{d}{p}) }{p^2\Gamma(1+\frac{d}{p})^2} \\
        =&\ \frac{d \cdot \Gamma(1+\frac{1}{p})^d}{p^2\Gamma(1+\frac{d}{p})} \cdot \left[\psi\left(1+\frac{d}{p}\right) - \psi\left(1+\frac{1}{p}\right)\right] \\
        =&\ \frac{d \cdot \Gamma(1+\frac{1}{p})^d}{p^2\Gamma(1+\frac{d}{p})} \cdot \left[\psi\left(\frac{d}{p}\right) + \frac{p}{d} - \psi\left(\frac{1}{p}\right) - p\right]
    \end{align*}
        by the general fact $\psi(1+x) = \psi(x) + \tfrac{1}{x}$
    \end{proof}
    Thus
    \begin{align*}
        \frac{\partial}{\partial p} W_{p}^{d}(k^{1/p}) =&\ k^{d/p} \cdot \frac{d \Gamma(1+\frac{1}{p})^d}{p^2\Gamma(1+\frac{d}{p})} \cdot \left[\psi\left(\frac{d}{p}\right) + \frac{p}{d} - \psi\left(\frac{1}{p}\right) - p\right] - \frac{dk^{d/p}\ln(k)}{p^2} \cdot \frac{\Gamma(1+\frac{1}{p})^d}{\Gamma(1+\frac{d}{p})} \\
        =&\ k^{d/p} \cdot \frac{d \Gamma(1+\frac{1}{p})^d}{p^2\Gamma(1+\frac{d}{p})} \cdot \left[\psi\left(\frac{d}{p}\right) + \frac{p}{d} - \psi\left(\frac{1}{p}\right) - p - \ln(k)\right].
    \end{align*}
    The first two terms in this product are always positive, so we continue by analyzing the third term, which we shorthand as $Q(d,p)$. We split into two cases for $k$. The following result, which is agnostic to $k$, will be useful in both.
    
    \begin{claim}{4.3.2}
    \label{claim:second_derivative}
    Let $d \geq 2$ and $p \geq 1$. Then
    \begin{equation*}
        \frac{\partial}{\partial p} Q(d,p) < 0.
    \end{equation*}
    \end{claim}
    \begin{proof}
        \begin{align*}
            \frac{\partial}{\partial p}\left[\psi\left(\frac{d}{p}\right) + \frac{p}{d} - \psi\left(\frac{1}{p}\right) - p\right] =&\ -\frac{d \cdot \psi'(d/p)}{p^2} + \frac{1}{d} + \frac{\psi'(1/p)}{p^2} - 1 \\
        =&\ \frac{1}{p^2}[\psi'(1/p) - d \cdot \psi'(d/p)] + \frac{1}{d} - 1.
        \end{align*}
        It is now enough to prove $\psi'(1/p) - d \cdot \psi'(d/p) < p^2(1 - \tfrac{1}{d})$ for $p \geq 1$. We employ the following bounds on $\psi'(x)$.
        \begin{claim}{4.3.3}[Theorem 1~\cite{GQZL15}]
        \label{claim:trigamma}
        For $x > 0$,
        \begin{equation*}
            \frac{1}{x+\frac{6}{\pi^2}} + \frac{1}{x^2} < \psi'(x) < \frac{1}{x+\frac{1}{2}} + \frac{1}{x^2}.
        \end{equation*}
        \end{claim}
        Applying \Cref{claim:trigamma} to upper bound $\psi'(1/p)$ and lower bound $\psi'(d/p)$ yields
        \begin{align*}
            \psi'(1/p) - d \cdot \psi'(d/p) <&\ \frac{1}{\frac{1}{p} + \frac{1}{2}} + p^2 - d\left(\frac{1}{\frac{d}{p} + \frac{6}{\pi^2}} + \frac{p^2}{d^2}\right) \\
            =&\ \frac{1}{\frac{1}{p} + \frac{1}{2}} - \frac{d}{\frac{d}{p} + \frac{6}{\pi^2}} + p^2\left(1-\frac{1}{d}\right).
        \end{align*}
        The final step is proving that the difference of the first two terms above is nonpositive. By
        \begin{equation*}
            \frac{1}{\frac{1}{p} + \frac{1}{2}} - \frac{d}{\frac{d}{p} + \frac{6}{\pi^2}} = \frac{1}{\frac{1}{p} + \frac{1}{2}} - \frac{1}{\frac{1}{p} + \frac{6}{\pi^2d}},
        \end{equation*}
        it suffices to have $\tfrac{6}{\pi^2d} \leq \frac{1}{2}$, or $d \geq \tfrac{12}{\pi^2} \approx 1.21$.
    \end{proof}
    
    With \Cref{claim:second_derivative} in hand, the two cases for $k$ are simple.
    
    \underline{Case 1}: $k > de^{\gamma-1}$, where $\gamma \approx 0.58$ is the Euler-Mascheroni constant. We use the upper bound $\psi(x) < -\tfrac{1}{x} + \ln(x+e^{-\gamma}) $~\cite{EGP00} at $p=1$ to rewrite $Q(d,p)$ as
    \begin{equation*}
        \psi(d) + \tfrac{1}{d} - \psi(1) - 1 - \ln(k) < \ln(d+e^{-\gamma}) + \gamma - 1 - \ln(k) \leq 0
    \end{equation*}
    by $\psi(1) = -\gamma$ and our assumption on $k$. It now suffices to prove that $\tfrac{\partial}{\partial p} Q(d,p)$ is negative, as this implies the minimum volume $\ell_p$ ball containing $V$ occurs at $p=\infty$. \Cref{claim:second_derivative} accomplishes this.
    
    \underline{Case 2}: $k \leq de^{\gamma-1}$. 
    If $k = 1$, the sum sampling shape is exactly the $l_1$ ball of radius 1. Suppose $k > 1$. Then
    \begin{equation*}
        Q(d, 1) = \psi(d) +\frac{1}{d} - \psi(1) - 1 - \ln(k) > \ln(d) + \gamma - 1 - \ln(k)
    \end{equation*}
    by the lower bound $\psi(x) > \ln(x) - \tfrac{1}{x}$~\cite[Equation 2.2]{A97}. This is nonnegative by our assumption on $k$. At $p=d$, the second term is instead
    \begin{align*}
        \psi(1) + 1 - \psi(1/d) - d - \ln(k) =&\ -[\psi(1/d) + d] - [\ln(k) + \gamma - 1] \\
        =&\ -\psi(1+1/d) - \ln(k) - \gamma + 1
    \end{align*}
    by  $\psi(1/d) = \psi(1+1/d) - d$. We know $\psi(x)$ increases from $\psi(1) = -\gamma$ to $\psi(2) = 1-\gamma$, so $d \geq 2$ implies
    \begin{equation*}
        -\psi(1+1/d)-\ln(k) - \gamma + 1 \leq -\ln(k) < 0.
    \end{equation*}
    $Q(d, p)$ is positive at $p=1$ and negative at $p=d$, so it suffices to show that it is monotonically decreasing in $p$, i.e., that its derivative with respect to $p$ is always negative. This implies that the minimum enclosing $\ell_p$ ball volume is minimized at either $p=1$ or $p=\infty$. \Cref{claim:second_derivative} therefore completes the result.
\end{proof}

It remains to show that the volume of $V$ is much smaller than that of the enclosing $\ell_1$ or $\ell_\infty$ ball for some values of $k$. We use the following result to bound the volume of $V$ at $k = d/e$. By \Cref{lem:R_volume}, the following lemma gives an estimate of the volume of $V$ in a single orthant, denoted $W_x$. Note that their statement is for volume normalized to a single orthant, which we maintain.

\begin{lemma}[Theorem 1~\citep{CKSS72}]
\label{lem:Eulerian number asymptotics}
    If $k = x\sqrt{\frac{d+1}{12}} + \frac{d+1}{2}$, then \\ $W_x = \lim_{d\rightarrow \infty}\sum_{j=1}^{k_{x,d}}\frac{A_{d,k_{x,d}}}{d!} = \frac{1}{\sqrt{2\pi}}\int_{-\infty}^{x}e^{-t^2/2}dt$.
\end{lemma}

The final statement follows.

\begin{lemma}
\label{thm:sum_rejection_sample}
    If $k = \frac{d}{e} - 1$, then rejection sampling $V$ using $kB_1^d$ or $B_\infty^d$ requires at least $C_{3}e^{C_{2}d}$ samples in expectation, where $C_3 > 0$ is independent of $d$.
\end{lemma}
\begin{proof}
    For $y > 0$,
    \begin{align*}
        \int_{-\infty}^{-y}e^{-t^2/2}dt =
        \int_{y}^{\infty}e^{-t^2/2}dt \leq 
        \frac{1}{y}\int_{y}^{\infty}te^{-t^2/2}dt =
        \frac{e^{-y^2/2}}{y}.
    \end{align*}
    Setting $x = ((\frac{2\sqrt{3}}{d+1})(\frac{d}{e}-1) - \sqrt{3})\sqrt{d+1}$ gives $k_{x,d} = \frac{d}{e} - 1$. Since we are taking a limit as $d\rightarrow \infty$, we can write $x \sim (\frac{2\sqrt{3}}{e} - \sqrt{3})\sqrt{d+1} \sim C\sqrt{d}$ where $C < 0$. Then since $x < 0$,
    \begin{align*}
        W_x \leq \frac{1}{\sqrt{2\pi}}\int_{-\infty}^{x}e^{\frac{-t^2}{2}}dt \leq \frac{e^{\frac{-x^2}{2}}}{|x|\sqrt{2\pi}} = \frac{e^{\frac{-(C^{2}(d+1))}{2}}}{-C\sqrt{2\pi d}} = \frac{C_{1}e^{-C_{2}d}}{\sqrt{d}}
    \end{align*}
    for some positive constants $C_1,C_2$ independent of $d$. 
    
    Since the single-orthant volume of the minimum enclosing $\ell_1$ ball of radius $\tfrac{d}{e}$ is $\tfrac{(d/e)^d}{d!} \sim \frac{1}{\sqrt{2\pi d}}$ by Stirling's approximation, the ratio of the volume between the sum sampling region and the $l_1$ ball of radius $d/e$ is $C_{3}e^{-C_{2}d}$ for some $C_3 > 0$ independent of $d$. Note that the $\ell_1$ ball of radius $d/e$ is indeed the lowest-volume $\ell_p$ ball, since the single-orthant volume of the minimum enclosing $\ell_\infty$ ball (of radius 1) is $1$.
\end{proof}
\section{Proofs For Count}
\label{sec:appendix_count}

\subsection{Proofs For Count Sampler}
We start by defining some terms that will repeatedly appear in the analysis. The following is an expanded version of \Cref{def:j_stuff} from the main body. Throughout, we often shorthand $\countball$ as $T$ for neatness and superscript the dimension $d$ when desired for emphasis.

\begin{definition}
\label{def:expanded_j_stuff}
    Let $J_{0}^{d} = (1,...,1)$ be the vector of $d$ 1s. Given $J \in \{-1,1\}^d$, we define:
    \begin{itemize}
        \item orthant $O(J_0^d) = \{x \in \mathbb{R}^d \mid x_1, \ldots, x_d \geq 0\}$ and orthant $O(J) = \{J * v: v \in O(J_{0}^{d})\}$ where $*$ is element-wise multiplication;
        \item $J_+, J_- \subseteq [d]$ are the sets of coordinates at which $J$ equals 1 and -1, respectively; and
        \item $T_{+}^{d} = \countball^{d} \cap O(J_{0}^{d})$ and $T_{-}^{d} = \countball^{d} \cap O(-J_{0}^{d})$ are the restrictions of $\countball$ to the positive and negative orthants, and $T^{d} = CH(T_{+}^{d} \cup T_{-}^{d}) = \countball^d$ is their convex hull.
    \end{itemize}
\end{definition}

We first determine the vertices $V_J$ of $\countball$ in an orthant indexed by $J \in \{-1,1\}^d$.

\characterizationVJ*
\begin{proof}
    $T^{d} = CH(T_{+}^{d} \cup T_{-}^{d})$, so its vertices are a subset of $V_{J_0} \cup V_{-J_0}$. Every vertex in $V_{J_0} \cup V_{-J_0}$ has all nonzero coordinates sharing a sign, so every vertex in $V_J \cap (V_{J_0} \cup V_{-J_0})$ has this property as well. $O(J)$ is the set of all points $p$ such that the positive coordinates of $p$ lie in $J_{+}$ and the negative coordinates of $p$ lie in $J_{-}$; call this property the sign condition of $J$. Then the elements of $V_J \cap (V_{J_0} \cup V_{-J_0})$ are the origin, vertices in $V_{J_0}$ with support contained in $J_+$, and vertices in $V_{-J_0}$ with support contained in $J_-$. It remains to show that there are no other vertices of $V_J$.
    
    Suppose $z \in V_J - (V_{J_0} \cup V_{-J_0})$. Then $z$ is not a vertex of $T^d$. Moreover, since $T = CH(T_+ \cup T_-)$, every point in $T_J$ lies on some line $L$ between distinct elements of $T_J \cap (T_+ \cup T_-)$ such that $L \subset T_J$. Therefore no vertex of $T_J$ can lie in the interior of $O(J)$. Define the standard bounding hyperplanes to be the $(d-1)$-dimensional subspaces that are orthogonal to the standard axes. We say that a shape $X$ fully intersects another shape $Y$ if the dimension of $X$ is equal to the dimension of $X \cap Y$. Then each of the $(d-1)$-dimensional standard bounding planes $P$ of $\partial O(J)$ fully intersects $T^{d}$ because $T^{d}$ contains a small ball $B$ around the origin and $P$ fully intersects $B$. In summary, $z$ lies on a $(d-1)$-dimensional polyhedron $S \subset P_{S} \cap T^{d}$ where $P_{S}$ is a bounding hyperplane of $\partial(O(J))$.
    
    Since vertices are extreme points, $z$ must be a vertex of $S$. Since $z$ is not in $V_{J_0} \cup V_{-J_0}$ and $z$ is a vertex of $S$, $z$ must be the interior of some edge $e = (v,w)$, where $v,w \in V_{J_0} \cup V_{-J_0}$, that intersects one of the standard bounding hyperplanes. To see this more explicitly, note that $z$ is not an extreme point of $T^{d}$, so there must be a small $j$-dimensional ball $b$, where $j \geq 1$, around $z$ such that $b \subset T^{d}$. If $j \geq 2$, then $P_{S} \cap b$ has dimension at least $j-1$ since at most one of the dimensions of $b$ can live in the one-dimensional complement of $P_{S}$. But then $P_{S} \cap b$ is a small ball of dimension at least 1 around $z$ in $S$, contradicting the fact that $z$ is an extreme point of $S$. So $j = 1$, or equivalently $z$ is an interior point of an edge $(v,w)$ where $v,w \in V_{J_0} \cup V_{-J_0}$. 
    
    If both $v$ and $w$ are in $V_{J_0}$ then each of their supports must be contained in $J_{+}$ or else a convex combination of $v$ and $w$ would have a positive value in a coordinate of $J_{-}$, violating the sign condition of $J$. Then $v,w \in O(J)$. If either $v$ or $w$ lie in the interior of $O(J)$, then the interior of $e$ lies in the interior of $O(J)$, contradicting the fact that $z$ lies on a standard bounding hyperplane of $O(J)$. It follows that both $v$ and $w$ lie on a standard bounding hyperplane of $O(J)$. If $v$ and $w$ lie on different bounding hyperplanes of $O(J)$, then the interior of $e$ once again lies in the interior of $O(J)$, leading to the same contradiction. But if $v$ and $w$ lie on the same bounding hyperplane of $O(J)$, then $v,w \in P_S$ since $z \in P_S$. Then $S$ contains $(v,w)$, so $z$ is not an extreme point of $S$, another contradiction. So it cannot be that $v,w$ are both in $V_{J_0}$, and similarly it cannot be that $v,w$ are both in $V_{-J_0}$.
    
    We can therefore assume that $v \in V_{J_0}$ and $w \in V_{-J_0}$. We take advantage of the fact that $(v,w)$ is an actual edge of $T^{d}$. This means that there exists a linear functional of the form $h:(x_1,...,x_d)\rightarrow(a_{1}x_{1} + ... + a_{d}x_{d})$, such that $h$ is maximized at $v$ and $w$ and at no other vertex of $T^{d}$. We say that $v$ and $w$ have a sign disagreement if there exists $1 \leq i \leq d$ where $v_i$ and $w_i$ have opposite sign.
    
    We show that $v$ and $w$ do not have a sign disagreement. Suppose they do, $v_i = 1$ and $w_i = - 1$. Since $h(v)$ is maximal, it must be that $a_i > 0$, or else we could construct the vertex $v' \in T^{d}$ formed from $v$ by zeroing out the $i$th coordinate, and then $h(v') \geq h(v)$. Similarly, since $h(w)$ is maximal, it must be that $a_i < 0$ or else we could construct the vertex $w' \in T^{d}$ formed from $w$ by zeroing out the $i$th coordinate, and then $h(w') \geq h(w)$. Since $a_i$ cannot be positive and negative simultaneously, this is a contradiction, so $v$ and $w$ have no sign disagreement. This means that the support of $v$ and $w$ are disjoint since $v$ has only positive non-zero coordinates and $w$ has only negative non-zero coordinates. Since $z$ is a convex combination of $v$ and $w$, and $z$ obeys the sign condition of $J$, it must be that the support of $v$ lies in $J_{+}$ and the support of $w$ lies in $J_{-}$. But then $v,w \in O(J)$, and we have previously shown this to be a contradiction.
\end{proof}

Next, we derive the volumes 
$|T(J)|$ of $\countball$ in different orthants. This involves reasoning about the faces of $\countball$ in different orthants.

\begin{definition}
\label{def:faces_etc}
    Let $T_{+,k}^d$ be the sum shape with ambient dimension $d$ and contribution $k$ restricted to the positive orthant $J_0$. Let $H_k$ be the hyperplane $x_1 + ... + x_d = k$. Let $G_0$ be the set of equations $\{x_i = 0\}_{i=1}^{d}$, and let $G_1$ be the set of equations $\{x_i = 1\}_{i=1}^{d}$. Index the faces $G$ of $[0,1]^{d}$ by $G_0 \cup G_1$.  
    
    Let $f$ be a map defined as follows. For each face $F \in G$, define $f(F)$ to be the set of points formed by starting with $F$ and deleting all points with $\ell_1$-norm larger than $k$. Then the faces of $T_{+,k}^d$ are $\{f(F): F \in G\} \cup ([0,1]^{d} \cap H_k)$.
\end{definition}

\begin{lemma}
\label{lem: characterization of G faces}
    If a face $F \in G_0$ is modified by $f$, then it is congruent to $T_{+,k}^{d-1}$. If a face $F \in G_1$ is modified by $f$, then it is congruent to $T_{+,k-1}^{d-1}$.
\end{lemma}
\begin{proof}
Any $F \in G_0$ is $(d-1)$-dimensional since one of its coordinates is constantly 0. The subset $Z$ of the rest of the coordinates are congruent to $[0,1]^{d-1}$ so if $F$ gets modified by the cutting plane $H_k$ as $F \rightarrow f(F)$ then $Z$ gets modified as $Z \rightarrow Z \cap H_k \sim T_{+,k}^{d-1}$. Similarly, a face $F \in G_1$ that is modified by $f$ has that $f(F) \sim T_{+,k-1}^{d-1}$ since the fixed coordinate  contributes 1 to the $\ell_1$ norm.
\end{proof}

The next result derives the (lower-dimensional) volume of the ``cut'' face of $T_+^d$ contained in the hyperplane $H_k$.

\begin{lemma}
\label{lem:count_weird_face}
    Let $\Delta_{d,k}$ be $[0,1]^{d} \cap H_k$. Then $|\Delta_{d,k}| = |R_{d-1,k}|\sqrt{d} = \frac{A_{d-1,k-1}}{(d-1)!}\sqrt{d}$.
\end{lemma}
\begin{proof}
    By the main result of \citet{CB74}, for any measurable set $Z$ in a $(d-1)$-dimensional affine subspace of $\mathbb{R}^d$, letting $\{\pi_j\}_{j=1}^d$ be the projection operations onto $x_j=0$,
    \begin{equation}
    \label{eq:shadows}
        |Z| = \sqrt{\sum_{j=1}^d |\pi_j(Z)|^2}.
    \end{equation}
    We briefly discuss some intuition for this result, starting with the special case of a parallelipiped $P$. The measure of $P$ is given by the square root of the Gram determinant of the matrix of vertices defining $P$, and we can compute this Gram determinant using the Cauchy-Binet formula to get the result. In the general case of a measurable set $Z$, we approximate $Z$ to arbitrary precision by covering it with little cubes and then show that applying the result for the parallelepiped to each cube individually and summing the resulting equations gives the desired general result.
    
    By \Cref{eq:shadows}, we can compute $|\Delta_{d,k}|$ by summing over the shadows in the $(d-1)$-dimensional subspaces that are orthogonal to the standard bases. The projection of $\Delta_{d,k}$ onto any one of these subspaces, say $x_j$ = 0, is congruent to $R_{d-1,k}$. This is because any point $x \in \Delta_{d,k}$ has $x_1 + ... + x_d = k$, so its projection onto $x_j = 0$ has $k-1 \leq x_1 + ... x_{j-1} + x_{j+1} + ... + x_d \leq k$. Then, recalling the definition of $R_{d-1,k}$ as the slice of the cube $[0,1]^{d-1}$ containing points with $\ell_1$ norm in $[k-1,k]$ and using \Cref{lem:R_volume},
    \begin{equation*}
        |\Delta_{d,k}| = \sqrt{\sum_{i=1}^{d}|\pi_{i}(\Delta_{d,k})|^{2}}
        = \sqrt{\sum_{i=1}^{d}|R_{d-1,k}|^{2}}
        = \frac{A_{d-1,k-1}}{(d-1)!}\sqrt{d}.
    \end{equation*}
\end{proof}

Next, we derive the volume of $T$ in each orthant $O(J)$.

\volumeOrthant*
\begin{proof}
    Let $V_{J,+} = V_J \cap V_{J_0}$ and $V_{J,-} = V_J \cap V_{-J_0}$. By \Cref{lem: characterization of V_J}, $V_J = V_{J,+} \cup V_{J,-}$, and $V_{J,+}$ is the set of vertices of $V_{J_0}$ with support contained in $J_+$, while $V_{J,-}$ is the set of vertices of $V_{-J_0}$ with support contained in $J_-$. Thus $CH(V_{J,+})$ is $T_{+,k}^{|J_+|}$ embedded at the coordinates of $J_+$ in the ambient space of $\mathbb{R}^d$, which is congruent to $T_{+,k}^j$. Similarly, $CH(V_{J,-}) \sim T_{+,k}^{d-j}$.
    
    We can think of every point $p \in T_J$ as belonging to a (not necessarily unique) convex combination of shapes, of the form $tCH(V_{J,+}) \oplus (1-t)CH(V_{J,-}) = tT_{+,k}^j \oplus (1-t)(-1)T_{+,k}^{d-j}$ for some $t \in [0,1]$. Let $t_p \in [0,1]$ be the smallest $t$ for which $p \in tT_{+,k}^j \oplus (1-t)(-1)T_{+,k}^{d-j}$. Define the shell $Y_{j,k}$ of $T_{+,k}^j$ to be the $(j-1)$-dimensional faces in $f(G_1)$ unioned with the cutting face $[0,1]^{j} \cap H_{k} = \Delta_{j,k}$. By the minimality of $t_p$ we know that the first summand factor of $p$ must be on the subset of the boundary of $t_{p}T_{+,k}^j$ since $t_{1}T_{+,k}^j \subset t_{2}T_{+,k}^j$ for $0 \leq t_1 < t_2 \leq 1$, i.e. $p \in t_{p}Y_{j,k} \oplus (1-t_{p})(-1)T_{+,k}^{d-j}$. We can therefore partition the points of $T_J$ into equivalence classes where $p$ is mapped to the class $t_p$.
    
    To see that each class $t \in [0,1]$ is nonempty, consider any point in $tCH(V_{J,+})$ that is a linear combination of the points of $V_{J,+}$ with no weight on the origin. Then we have the disjoint union $T_J = \sqcup_{t \in [0,1]}tY_{j,k} \oplus (1-t)(-1)T_{+,k}^{d-j}$, and we can set up the integral
    \begin{align*}
        |T_J| &= \int_{0}^{1}|tY_{j,k} \oplus (1-t)(-1)T_{+,k}^{d-j}|\partial t \\ &= \int_0^1 |tY_{j,k}||(1-t)T_{+,k}^{d-j}|\partial t.
    \end{align*}
    We then compute the shell volume $|tY_{j,k}|$ by interpreting it as the rate of change of the volume of $|tT_{+,k}^{j}|$
    \begin{equation*}
        |tY_{j,k}| = \frac{\partial}{\partial t}|tT_{+,k}^j| = \frac{\partial}{\partial t}\left(t^{j}|T_{+,k}^j|\right) = \frac{\partial}{\partial t}\left(t^{j}\sum_{i=1}^{k}|R_{j,i}|\right) = jt^{j-1}\sum_{i=1}^{k}\left(\frac{A_{j,i-1}}{j!}\right)
    \end{equation*}
    where $R_{j,i}$ is a slice of the cube $[0,1]^j$ containing points with $\ell_1$ norm in $[i-1,i]$, per \Cref{lem:R_volume}. Continuing the integral
    \begin{align*}
        |T_J| &= \int_{0}^{1}\left[jt^{j-1}\sum_{i=1}^{k}\frac{A_{j,i-1}}{j!}\right]\left[(1-t)^{d-j}|T_{+,k}^{d-j}|\right]\partial t \\
        &= j\left(\sum_{i=1}^{k}\frac{A_{j,i-1}}{j!}\right)\left(\sum_{i=1}^{k}\frac{A_{d-j,i-1}}{(d-j)!}\right)\int_{0}^{1}t^{j-1}(1-t)^{d-j}\partial t \\
        &= \left(\sum_{i=1}^{k}\frac{A_{j,i-1}}{j!}\right)\left(\sum_{i=1}^{k}\frac{A_{d-j,i-1}}{(d-j)!}\right)\left(\frac{j}{d}\right)\binom{d-1}{j-1}^{-1} \\
        &= \left(\sum_{i=1}^{k}A_{j,i-1}\right)\left(\sum_{i=1}^{k}A_{d-j,i-1}\right)\frac{1}{d!}
    \end{align*}
    where the third equality follows from repeated integration by parts. To see that, let $f(j) = \int_{0}^{1}x^{j-1}(1-x)^{d-j}dx$. Then setting $u(x) = x^{j-1}$ and $v(x) = -\frac{1}{d-j+1}(1-x)^{d-j+1}$ lets us write
    \begin{equation*}
        f(j) = [u(x)v(x)]_{0}^{1} + \frac{j-1}{d-j+1} \int_{0}^{1}x^{j-2}(1-x)^{d-j+1}dx = \frac{j-1}{d-j+1}f(j-1)
    \end{equation*}
    until
    \begin{equation*}
        f(1) = \int_0^1 (1-t)^{d-1}dt = \left(-\frac{(1-t)^d}{d}\right]_0^1 = \frac{1}{d}.
    \end{equation*}
\end{proof}

The next result shows how to draw a uniform random sample from $T(J)$, the restriction of $\countball$ to orthant $O(J)$.

\begin{lemma}
\label{lem: sampling from orthant}
    Let $J \in \{-1,1\}^{d}$ correspond to an orthant. Suppose $|J_{+}| = j$ and $|J_{-}| = d - j$. Sampling from $T(J)$ reduces to sampling $\betad{j}{d-j+1}$ and then calling the Sum sampler (\Cref{alg:sum_sampler}) twice. In total, this takes time $O(d^2)$.
\end{lemma}
\begin{proof}
    By \Cref{eq:beta}, as derived in the proof of \Cref{lem: volume of orthant}, the cross sections of $V_J$, for $t \in [0,1]$, have volume proportional to $t^{j-1}(1-t)^{d-j}dt$. We can therefore pick a cross section $t \in [0,1]$ by sampling $\betad{j}{d-j+1}$.
    
    It then remains to sample a point from the cross-section $tY_{j, k} \oplus (1-t)(-1)T_{+,k}^{d-j}$. Recall from the definition of $Y_{j,k}$ and \Cref{lem: characterization of G faces} that $Y_{j,k}$ contains $j$ shapes congruent to $tT_{+,k-1}^{j-1}$, and if $k < j$, then $Y_{j,k}$ additionally contains one shape congruent to $t\Delta_{j,k}$ (\Cref{lem:count_weird_face}). We will sample from $tY_{j,k}$ by defining weights proportional to the volumes of the sub-shapes of $Y_{j,k}$.
    
    If $j = 1$, then $tY_{1,k} = \{t\}$ and we are done. If $j > 1$, define function $q(t) = t$ to measure the perpendicular distance between the $x_i = t$ and $x_i=0$ planes. Using the fact that $\frac{\partial q}{\partial t} = 1$, we define weights for the $j$ shapes congruent to $tT_{+,k-1}^{j-1}$:
    \begin{align*}
        w_1 = ... = w_j =&\ |tT_{+,k}^{j-1}\partial q| \\
        =&\ t^{j-1}\sum_{i=1}^{k-1}|R_{j-1,i}|\partial t \\
        =&\ \tfrac{t^{j-1}}{(j-1)!}\left(\sum_{i=1}^{k-1}A_{j-1,i-1}\right)\partial t
    \end{align*}
    by \Cref{lem:R_volume}, noting that we apply it with $d=j-1$. Additionally, if $k < j$, then we need to define a weight $w_{j+1}$ for the $\Delta_{j,k}$ face. Define function $s(t) = t \cdot \tfrac{k}{\sqrt{j}}$ to be the perpendicular distance from the plane $tH_k$ (containing $\Delta_{j,k}$) to $H_0$. Then since $\tfrac{\partial s}{\partial t} = \tfrac{k}{\sqrt{j}}$,
    \begin{align*}
        w_{j+1} =&\ \left\lvert t\Delta_{j,k}\partial s\right\rvert \\
        =&\ \left\lvert t\Delta_{j,k}\left(\frac{k\partial t}{\sqrt{j}}\right)\right\rvert \\
        =&\ t^{j-1}\left\lvert R_{j-1,k}\sqrt{j}\left(\frac{k\partial t}{\sqrt{j}}
        \right)\right\rvert \\
        =&\ t^{j-1}\left\lvert R_{j-1,k}k\partial t\right\rvert \\
        =&\ \tfrac{t^{j-1}k}{(j-1)!}A_{j-1,k-1}\partial t
    \end{align*}
    where we have used \Cref{lem:count_weird_face} and the fact that scaling a $(j-1)$-dimensional object $\Delta_{j,k}$ by $t$ changes its measure by a factor of $t^{j-1}$.
    
    After selecting one of the indices $i \in \{1,...,j+1\}$ via the normalized $w_i$ weights, if $1 \leq i \leq j$ then we sample a point $p_1 \in T_{+,k-1}^{j-1}$ by calling the Sum sampler (\Cref{alg:sum_sampler}). If $i = j+1$ we can sample a point $p_1 \in \Delta_{j,k} \sim R_{j-1,k}$ (the isomorphism from \Cref{lem:count_weird_face} induced by forgetting the last coordinate) by calling the portion of the Sum sampler that samples from a particular $R$ slice (\Cref{alg:sum_sampler}). In either case, we sample a point $p_2 \in T_{+,k}^{d-j}$ using the Sum sampler. Finally, let $y_1$ be formed starting with the all zeros vector by embedding $tp_1$ at $J_{+}$, and let $y_2$ be formed starting with the all zeros vector by embedding $(1-t)(-1)p_2$ at $J_{-}$. Then $y_1 \oplus y_2 \in tY_{j,k} \oplus (1-t)(-1)T_{+,k}^{d-j}$ is a point uniformly sampled from the $t$ cross section of $V_J$.
    
    Sampling a Beta distribution takes time $O(d)$, and each call to the Sum sampler costs $O(d^2)$. This yields overall time $O(d^2)$.
\end{proof}
The last step is putting these results together to obtain the final algorithm (\Cref{alg:count_sampler}) and guarantee.
\begin{theorem}
\label{thm:count_sampler}
There is an algorithm to sample a point from $\countball$ in time $O(d^2)$.
\end{theorem}
\begin{proof}
    The first step is to pick an orthant $J \in \{-1,1\}^{d}$. Suppose $|J_{+}| = j$ and $|J_{-}| = d - j$. There are $\binom{d}{j}$ orthants $J'$ where $T_J$ is isometric to $T_{J'}$. Let $\{C_0,C_1,...,C_d\}$ be the equivalence classes of orthants partitioned by isometry where each orthant $J \in C_j$ has $|J_{+}| = j$ and $|J_{-}| = d - j$. For $0 \leq j \leq d$, let $z_{j}'$ be the total volume of the orthants in $C_j$, and let $z_{j}$ be the normalized $z_{j}'$ weights. By \Cref{lem: volume of orthant},
    \begin{align*}
        z_{j}' =&\ \frac{d!}{j!(d-j)!} \left(\sum_{i=1}^k A_{j,i-1}\right)\left(\sum_{i=1}^k A_{d-j,i-1}\right)\frac{1}{d!} \\
        =&\ \left(\sum_{i=1}^{k}\frac{A_{j,i-1}}{j!}\right)\left(\sum_{i=1}^{k}\frac{A_{d-j,i-1}}{(d-j)!}\right).
    \end{align*}
    After computing the table of Eulerian numbers up to the row $d$ (time $O(d^2)$), we can make one pass across rows $j$ and $d-j$ to compute the partial sums required for $z_{j}'$ (time $O(d)$). Thus, computing the $z_j$ weights costs $O(d^2)$ overall. 

    We can therefore choose an orthant by sampling a class $C_j$ with weight $z_j$ and then choosing a random vector with $j$ 1s and $d-j$ $-1$s, which takes time $O(d)$, so picking a random orthant takes $O(d^2)$. After choosing an orthant, we sample a point uniformly from it by \Cref{lem: sampling from orthant} in $O(d^2)$.
\end{proof}

\begin{algorithm}
    \begin{algorithmic}[1]
    \STATE {\bfseries Input:} Dimension $d$ and $\ell_0$ bound $k$
    \STATE Compute the $\{z_0,...,z_{d}\}$ weights corresponding to $\{C_0,...,C_{d}\}$ using \Cref{thm:count_sampler}
    \STATE Sample a class $C_j$ according to the $z$ weights
    \STATE Sample an orthant $J \in C_j$
    \STATE Sample cross section index $t \sim \betad{j}{d-j+1}$
    \STATE Compute the $\{w_1,...,w_{j}\}$ weights using \Cref{lem: sampling from orthant}
    \IF{$k < j$}
        \STATE Compute weight $w_{j+1}$ using \Cref{lem: sampling from orthant}
    \ENDIF
    \STATE Sample cross section face index $i$ according to the $w$ weights
    \STATE If $1 \leq i \leq j$, sample point $p_1 \in T_{+,k-1}^{j-1}$ using the Sum sampler (\Cref{alg:sum_sampler})
    \IF{$i=j+1$}
        \STATE Sample point $q \in R_{j-1,k}$ by \Cref{alg:sum_sampler}
        \STATE Let $q_{j} = k - \sum_{i=1}^{j-1}q_{i}$
        \STATE Define uniform sample $p_1 = q \oplus q_{j} \in \Delta_{j,k}$ using the isomorphism from \Cref{lem:count_weird_face}
    \ENDIF
    \STATE Sample point $p_2 \in T_{+,k}^{d-j}$ by a call to sum sampler \Cref{alg:sum_sampler}
    \STATE Define $y_1$ by embedding $tp_1$ at $J_{+}$ in the all zeros vector of length $d$
    \STATE Define $y_2$ by embedding $(1-t)(-1)p_2$ at $J_{-}$ in the all zeros vector of length $d$
    \STATE Return $y_1 \oplus y_2$
    \caption{Count Sampler}
    \label{alg:count_sampler}
    \end{algorithmic}
\end{algorithm}

\subsection{Proofs For Count Ellipse}

\ellipseCentering*
\begin{proof}
    Suppose not. Let $E$ be a minimum ellipse of $\countball$ that is not origin-centered. Let $U$ be the unique linear operator that maps $E$ to a $(d-1)$-dimensional unit sphere. Linear transformations preserve symmetry around the origin, so $U(\countball)$ is origin-centered, and $U(E)$ is not. For any point $x \in U(\countball)$, $U(\countball)$ contains the line segment $(x,-x)$ of length $2\|x\|_2$, so $U(E)$ encloses it as well. $U(E)$ is a sphere of radius 1, so if $\|x\|_2=1$, $U(E)$ can only enclose $(x,-x)$ by being origin-centered, a contradiction. It follows that $U(\countball)$ lies in an origin-centered ball of radius $R < 1$.
    
    Let $E_c = E - v$ be the ellipse formed by translating the center of $E$ to the origin. First, we show that $E_c$ has smaller average squared $\ell_2$ norm than $E$. Let $p(X)$ be a point sampled uniformly from the space enclosed by some ellipse $X$. Then
    \begin{align*}
        \E{}{\|p(E)\|_2^2} =&\ \E{}{\|p(E_c+v)\|_2^2} \\
        =&\ \E{}{p(E_c + v)^Tp(E_c+v)} \\
        =&\ \E{}{(v+p(E_c))^T(v+p(E_c))} \\
        =&\ \|v\|_2^2 + 2v^T\E{}{p(E_c)} + \E{}{\|p(E_c)\|_2^2} \\
        =&\ \|v\|_2^2 + \E{}{\|p(E_c)\|_2^2}
    \end{align*}
    so $\mathbb{E}[\|p(E)\|_{2}^{2}] > \mathbb{E}[\|p(E_c)\|_{2}^{2}]$. 
    
    Finally, we show that $E_c$ still contains $\countball$. Note that $U(E_c) = U(E - v) = U(E) - U(v)$ and since $U(E)$ is a unit sphere, then $U(E_c)$ is a translated unit sphere. Furthermore, since $E_c$ is origin-centered, then any linear transformation of $E_c$ is also origin-centered. Then $U(E_c)$ is an origin-centered unit sphere. Since $U(\countball)$ lies in an origin-centered ball of radius $R < 1$, $U(\countball) \subset U(E_c)$, and applying $U^{-1}$ over this statement gives that $\countball \subset E_c$. But then $E_c$ is a ``more optimal'' ellipse than $E$, a contradiction.
\end{proof}

\ellipseNorm*
\begin{proof}
    We first analyze the expected squared $\ell_2$ norm of a sample from $E_A$ itself. Let $V = \{v_1,...,v_d\}$ be an orthonormal basis of eigenvectors of $A$. Let $X$ be a uniform sample from the sphere defined by the equation $x_{1}^{2} + ... + x_{d}^{2} = 1$ where $(x_1,...,x_d)$ is written in the $V$ basis. Let $Y$ be a uniform sample from $E_A$. We can draw $Y$ by sampling a uniformly random point on the unit sphere and then scaling the directions of the eigenvectors of $A$ by the axes lengths $a_i$. This procedure produces a uniform sample from $E_A$ because the above scaling is a linear transformation. Then $Y = \sum_{i=1}^{d}a_{i}X_{i}v_{i}$ where $X_i$ is the random variable for the $i^{th}$ coordinate, and
    \begin{align*}
        \mathbb{E}[\|Y\|_{2}^{2}] &= \mathbb{E}\left[\left(\sum_{i=1}^{d}a_{i}X_{i}v_{i}\right)^{T}\left(\sum_{i=1}^{d}a_{i}X_{i}v_{i}\right)\right] \\
        &= \mathbb{E}\left[\sum_{i=1}^{d}\sum_{j=1}^{d}a_{i}a_{j}X_{i}X_{j}v_{i}^{T}v_{j}\right] \\
        &= \mathbb{E}\left[\sum_{i=1}^{d}a_{i}^{2}X_{i}^{2}v_{i}^{T}v_{i}\right] \\
        &= \sum_{i=1}^{d}a_{i}^{2}\mathbb{E}[X_{i}^{2}] \\
        &= \frac{1}{d}\sum_{i=1}^{d}a_{i}^{2}.
    \end{align*}
    
    We now analyze $Z$, a sample from $\enc{E_A}$. Let $\Omega_{d} = \frac{\pi^{d/2}}{\Gamma(1+d/2)}$ be the volume of the unit ball. Then $|tE_A| = t^{d}\Omega_{d}\prod_{j=1}^{d}a_j$, and $\frac{\partial |tE|}{\partial t} = dt^{d-1}\Omega_{d}\prod_{j=1}^{d}a_j$. For $t \in [0,1]$, let $L_t$ be the expected squared $\ell_2$ norm of a uniform sample from the $t^{th}$ ellipse shell $\partial(tE_A)$. By the above analysis of $Y$, $L_t = \frac{1}{d}\sum_{j=1}^{d}(ta_j)^{2} = \frac{t^{2}}{d}\left(\sum_{j=1}^{d}a_j^{2}\right)$. The density for a small neighborhood of $\partial(tE_A)$ is $p_t = \frac{\partial|tE_A|}{|E_A|} = dt^{d-1}\partial t$. Then
    \begin{equation*}
        \E{}{\|Z\|_2^2} = \int_{0}^{1} L_{t}p_{t} = \int_{0}^{1} \frac{t^{2}}{d}\left(\sum_{j=1}^{d}a_j^{2}\right)dt^{d-1}\partial t = \left(\sum_{j=1}^{d}a_j^{2}\right)\int_{0}^{1}t^{d+1}\partial t = \frac{\sum_{j=1}^da_j^2}{d+2}.
    \end{equation*}
\end{proof}

\ellipseUnique*
\begin{proof}
    Suppose we have minimum ellipses $E_A$ and $E_B$. We argue that the ``average" ellipse given by $x^{T}(A+B)x = 2$ has a lower expected squared $\ell_2$ norm than $E_A$, a contradiction. Note that this average ellipse would automatically contain $\countball$ since points that satisfy both equations separately will satisfy the sum of the two equations.
    
    By \Cref{lem:centering_lemma}, $E_A$ and $E_B$ are origin-centered, so we can apply \Cref{lem:ellipse_squared_l2} to relate their average squared $\ell_2$ norms to their squared axes lengths. By \Cref{def:ellipse}, the squared axes lengths of $A$ are equal to the reciprocals of their eigenvalues, and the same holds for $B$. It therefore suffices to show that $A+B$ has smaller sum of reciprocal eigenvalues than that of $A$ to reach a contradiction. To analyze the eigenvalues of $A+B$, we apply the Courant-Fischer theorem.
    
    \begin{lemma}[Courant-Fischer]
    \label{lem:Courant-Fischer}
        Let $M$ be a real symmetric positive definite $d \times d$ matrix with eigenvalues $0 < \lambda_1(M) \leq \cdots \leq \lambda_d(M)$. Then for each $j \in [d]$, 
        \begin{equation}
            \lambda_j(M) = \min\{\max\{R_M(x) \mid x \in U, x \neq 0\} \mid \dim(U) = j\} 
        \end{equation}
        where $R_M(x) = \frac{x^TMx}{x^Tx}$.
    \end{lemma}
    We have $R_{A+B}(x) = \frac{x^{T}(A+B)x}{x^{T}x} = \frac{x^{T}Ax}{x^{T}x} + \frac{x^{T}Bx}{x^{T}x} = R_{A}(x) + R_{B}(x)$. Because $A, B, A+B$ are positive definite, $R_{A}(x),R_{B}(x)$, and $R_{C}(x)$ are positive. Thus $R_{A+B}(x) > \max\{R_{A}(x),R_{B}(x)\}$, so by \Cref{lem:Courant-Fischer}, $\lambda_{j}(A+B) > \max\{\lambda_j(A), \lambda_j(B)\}$, and $\sum_{j=1}^d \frac{1}{\lambda_j(A+B)} < \sum_{j=1}^d \frac{1}{\lambda_j(A)}$.
\end{proof}

\ellipseAxes*
\begin{proof}
    Let $\sigma_{i,j}$ be the reflection that switches coordinates $i$ and $j$. To see that it's a reflection, let $\Pi_{i,j}$ be the hyperplane that passes through $\{e_1,...,e_d\} - \{e_i, e_j\}$, the point $\frac{1}{2}(e_i + e_j)$, and the origin. Then $\sigma_{i,j}$ is the operator whose action is reflection across $\Pi_{i,j}$. Since $\sigma_{i,j}$ is an isometry that fixes the origin, the expected squared $\ell_2$ norm of points enclosed by $\sigma_{i,j}(E)$ is equal to that of $E$. By \Cref{lem:count_ellipse_unique}, $\sigma_{i,j}(E) = E$. 
    
    This means that $E$ has reflection symmetry over $\Pi_{i,j}$. We show that the orthonormal vector $v_{i,j}$ to $\Pi_{i,j}$ is an eigenvector, and thus a valid axis direction of $A$.
    \begin{claim}
    \label{claim:reflection_eigenvectors}
        Let $E$ be an ellipse with associated matrix $A$. If $w$ is a vector pointing from the origin to a point in $E$, and $w$ orthogonal to a hyperplane $H_w$ for which $E$ has reflection symmetry, then $w$ is an eigenvector of $A$.
    \end{claim}
    \begin{proof}
        We use induction on the dimension $d$. Let $w'$ be a vector orthogonal to $w$, and define basis $\{w, w',u_1,...,u_{d-2}\}$ of $\mathbb{R}^{d}$. Let $E'$ be the ellipse that is the image of $E$ under the linear map $p$ where $p(w') = 0$ and $p$ is the identity map on $w', u_1, \ldots, u_{d-2}$, and let $\pi$ be the reflection operator where $\pi(w) = -w$ and $\pi$ is the identity map on $w', u_1, \ldots, u_d$.
        
        Since $v \in E$ implies $\pi(v) \in E$, applying $p$ over this statement gives that $p(v) \in p(E)$ implies $p(\pi(v)) \in p(E)$. Since $p(E) = E'$ and $p$ and $\pi$ commute, we can write this as $p(v) \in E'$ implies $\pi(p(v)) \in E'$. In other words, $\pi$ is a reflection operator for $E'$, and $p(w) = w$ is the reflecting vector for $\pi$ in $E'$.
        
        Let $A_{w'}$ be the restriction of $A$ to the orthogonal complement of $\mbox{span}(w')$. Since $E'$ has dimension one less than $E$, by the inductive hypothesis, as the reflecting vector for $\pi$, $w$ is an eigenvector of $A_{w'}$ where $w$ is viewed as living in the domain of $A_{w'}$. But since $A_{w'}$ is a restriction of $A$, then $w$ is also an eigenvector of $A$ when $w$ is viewed as living in the domain of $A$.
    
        It remains to show the base case of $d = 2$. Let $\{v_1, v_2\}$ be orthonormal eigenvectors of $A$ with eigenvalues $\{\lambda_1, \lambda_2\}$. Write $w = c_{1}v_{1} + c_{2}v_{2}$, and define $w' = c_{2}v_{1} - c_{1}v_{2}$. Then $\{w,w'\}$ is an orthogonal basis. Let $v = b_{1}w + b_{2}w'$ be a point on $E$. Then since $E$ has reflection symmetry over $w$ and is defined by $x^TAx=1$, we have the following two equalities:
        \begin{align*}
            (b_{1}w + b_{2}w')^{T}A(b_{1}w + b_{2}w') &= 1 \\
            (-b_{1}w + b_{2}w')^{T}A(-b_{1}w + b_{2}w') &= 1
        \end{align*}
        Subtracting and simplifying these gives $w'^{T}Aw + w^{T}Aw' = 0$. Since $A$ is positive definite, it induces the inner product defined by $(x,y)_{A} \rightarrow x^{T}Ay$. Inner products pairings are symmetric, so $w'^{T}Aw = w^{T}Aw'$, so $2w^{T}Aw' = 0$ and $w^{T}Aw' = 0$. Expanding the last equation gives $(c_{1}v_{1} + c_{2}v_{2})^{T}(c_{2}\lambda_{1}v_{1} - c_{1}\lambda_{2}v_{2}) = 0$, and since the cross terms are zero this simplifies to $c_{1}c_{2}\lambda_{1}v_{1}^{T}v_{1} - c_{1}c_{2}\lambda_{2}v_{2}^{T}v_{2}$ or $c_{1}c_{2}(\lambda_{1} - \lambda_{2}) = 0$. If $c_{1} = 0$ or $c_{2} = 0$, then we are done as $w$ is a scaling of eigenvector $v_1$ or $v_2$ and so is an eigenvector as well. Otherwise $\lambda_{1} = \lambda_{2}$ which means that $E$ is a circle, so $A$ is a multiple of the identity and every vector is an eigenvector. In particular, $w$ is an eigenvector.
    \end{proof}
    
    $E$ has reflection symmetry over $\Pi_{1,j}$ for $2 \leq j \leq d$, so each element $v_{1,j}$ of of $\{v_{1,2},...,v_{1,d}\}$ corresponds to an eigenvector of $A$ with eigenvalue $a_j$. We show that no pair among $\{v_{1,2},...,v_{1,d}\}$ is orthogonal. Each $\Pi_{1,j}$ orthogonally bisects the edge between $e_1$ and $e_j$, so the direction of $v_{1,j}$ is parallel to the vector $e_j - e_1$; however, there are no pairs of orthogonal edges among $\{e_2 - e_1,...,e_d - e_1\}$ since $(e_i - e_1)^{T}(e_j - e_1) = 1$ for all $2 \leq i < j \leq d$. Since the eigenspaces of symmetric PSD matrices (the class of matrices containing $A$) are orthogonal, all of these principal axes correspond to the same eigenvalue. In other words, $a_2 = ... = a_d$.
    
    It remains to determine the final eigenvector with eigenvalue $a_1$. If $a_1 = a_2$, then $A$ is a multiple of the identity, so in particular $(1,...,1)$ is an eigenvector of $A$. If $a_1 \neq a_2$, then the final eigenvector must be orthogonal to each $v_{1,j}$ since distinct eigenspaces are orthogonal. The $(1,...,1)$ vector spans the orthogonal complement of $\mbox{span}(v_{1,2},...,v_{1,d})$ since $v_{1,j}^{T}(1,...,1) = (e_j - e_1)^{T}(1,...,1) = 0$, so $(1,...,1)$ is the final eigenvector.
\end{proof}

\ellipseContacts*
\begin{proof}
    Define $v_1(j) = \tfrac{j}{d}(1, 1, \ldots, 1)$ and let $v_2(j)$ be a vector that points from $v_1(j)$ to an arbitrary point with $j$ 1s and $d-j$ 0s. Then $v_2(j)$ consists of $j$ coordinates with $\tfrac{d-j}{d}$ and $d-j$ coordinates with $-\tfrac{j}{d}$, so $v_{2}(j)$ is orthogonal to $v_{1}(j)$, and
    \begin{equation*}
        \|v_{2}(j)\|_2 = \sqrt{j\left(1-\frac{j}{d}\right)^{2} + (d-j)\left(\frac{j}{d}\right)^{2}} = \sqrt{j - \frac{2j^2}{d} + \frac{j^3}{d^2} + \frac{j^2}{d} - \frac{j^3}{d^2}} = \sqrt{\frac{j(d-j)}{d}}.
    \end{equation*}
    The expression inside the root is a down-ward facing parabola maximized at $j=d/2$. The minimum ellipse has an axis along $(1, 1, \ldots, 1)$ (\Cref{lem:axes_directions}), must contact vertices of $\countball$ by its minimality, and has ellipse cross-section radius decreasing away from the origin. Therefore if the ellipse intersects any of the points $v_{1}(j) + v_{2}(j)$ where $0 < j < k$, then it does not enclose $v_{1}(k) + v_{2}(k)$ since $\|v_{2}(j)\|_2$ is increasing for $0 \leq j \leq d/2$, a contradiction.
\end{proof}

\ellipseThm*

\begin{proof}
     By \Cref{lem:centering_lemma} and \Cref{lem:ellipse_squared_l2}, to compute $E$ with axes lengths $a_1, \ldots, a_d$, we minimize objective function $\sum_{j=1}^d a_j^2$. By \Cref{lem:axes_directions}, this reduces to $a_1^2 + (d-1)a_2^2$. Let $e_{v_1} = \frac{1}{\sqrt{d}}(1,...,1)$, and let $e_{v_2} = \frac{1}{\sqrt{2}}(-1, 1, 0,...,0)$. Extend $\{e_{v_1}, e_{v_2}\}$ to a full orthonormal basis $B = \{e_{v_1},...,e_{v_d}\}$. By \Cref{lem:ellipse_contacts}, $k \leq d/2$ means that the minimum ellipse $E$ intersects $v_{1}(j) + v_{2}(j) = \|v_{1}(j)\|_2e_{v_1} + \|v_{2}(j)\|_2e_{v_2}$ which is written as $(\|v_{1}(j)\|_2, \|v_{2}(j)\|_2,0,...,0)$ in the $B$ basis.
     
     Consider the program whose objective function is $f(a_1, a_2) = a_{1}^{2} + (d-1)a_{2}^{2}$, and whose constraint in the $B$ basis can be written as $g(a_1, a_2) = \frac{\|v_{1}(k)\|_2^{2}}{a_{1}^{2}} + \frac{\|v_{2}(k)\|_2^{2}}{a_{2}^{2}} - 1 = 0$. Define the Lagrangian $\mathcal{L}(a_1, a_2, \lambda) = f(a_1, a_2) + \lambda g(a_1, a_2)$. Any optimal point of $\mathcal{L}$ satisfies that $\nabla\mathcal{L} = 0$, so calculating yields
    \begin{align*}
        \frac{\partial\mathcal{L}}{\partial a_1} = 2a_{1} - 2\lambda\frac{\|v_{1}(k)\|_2^{2}}{a_{1}^{3}} &= 0 \\ a_1 &= \left(\lambda\|v_{1}(k)\|_2^{2}\right)^{1/4} \\
        \frac{\partial\mathcal{L}}{\partial a_2} = 2(d-1)a_{2} - 2\lambda\frac{\|v_{2}(k)\|_2^{2}}{a_{2}^{3}} &= 0
        \\ a_2 &= \left(\frac{\lambda\|v_{2}(k)\|_2^{2}}{d-1}\right)^{1/4} \\
        \frac{\partial\mathcal{L}}{\partial \lambda} = g(a_1, a_2) &= 0 \\
        \frac{\|v_{1}(k)\|_2^{2}}{(\lambda\|v_{1}(k)\|_2^{2})^{1/2}} + \frac{\|v_{2}(k)\|_2^{2}}{\left(\frac{\lambda\|v_{2}(k)\|_2^{2}}{d-1}\right)^{1/2}} - 1 &= 0 \\
        \frac{\|v_{1}(k)\|_2}{\sqrt{\lambda}} + \frac{\|v_{2}(k)\|_2\sqrt{d-1}}{\sqrt{\lambda}} - 1 &= 0 \\ \lambda &= (\|v_{1}(k)\|_2 + \|v_{2}(k)\|_2\sqrt{d-1})^{2}
    \end{align*}
    and we plug in $\|v_1(k)\|_2 = \frac{k}{\sqrt{d}}$ and $\|v_2(k)\|_2 = \sqrt{k(d-k)/d}$ from the proof of \Cref{lem:ellipse_contacts} to get
    \begin{align*}
        \lambda = \left(\frac{k}{\sqrt{d}} + \sqrt{\frac{k(d-k)(d-1)}{d}}\right)^2 =&\ \frac{k}{d}\left(\sqrt{k} + \sqrt{(d-k)(d-1)}\right)^2 \\
        a_1 =&\ \left(\frac{\lambda k^2}{d}\right)^{1/4} \\
        a_2 =&\ \left(\frac{\lambda k (d-k)}{d(d-1)}\right)^{1/4}.
    \end{align*}
\end{proof}
\section{Proofs For Vote}
\label{sec:appendix_vote}

\subsection{Proofs For Vote Sampler}
We start with a result characterizing the edges of $CH(P_d)$.

\begin{lemma}[\cite{GG77}]
\label{lem:transposition}
    For a fixed vertex $(v_1,...,v_d) \in CH(P_d)$, each neighboring vertex is formed by picking a value $i \in \{0,...,d-2\}$ and then switching the coordinate containing value $i$ and the coordinate containing value $i+1$.
\end{lemma}

Next, we prove the full version of \Cref{lem:simple_face_decomposition}, originally given without proof as Proposition 2.6 of~\citet{P09}.

\begin{lemma}
\label{lem:face_decomposition}
    Given integer $k$ such that $0 \leq k \leq d-1$, there is a bijection between the $k$-dimensional faces of $CH(P_d)$ and the collection of sequences of $d-k$ subsets partitioning $[d]$. Let $T_1$ be the top $|B_1|$ elements of $\{0,...,d-1\}$. For $2 \leq i \leq d-k$ let $T_i$ be the top $|B_i|$ elements of $\{0,...,d-1\} - \cup_{j=1}^{i-1}T_j$. If $F$ is a $k$-dimensional face of $CH(P_d)$ corresponding to subsets $B_1,...,B_{d-k}$, then $F$ has direct sum decomposition $\oplus_{i=1}^{d-k}(CH(P_{B_i}) + \min(T_i)I_{B_i})$.
\end{lemma}
\begin{proof}
    Let $F$ be a $k$-dimensional face of $CH(P_d)$, and let $V_F$ be the vertices of $F$. Let $\{v_1,...,v_{d-k}\}$ be $d-k$ linearly independent vectors such that each $v_i$ is orthogonal to $F$. Since $\dim(F) = k$, there exist $d-k$ relations $r = \{r_1,...,r_{d-k}\}$ where $r_i$ is $v_i \cdot x = c_i$. Every vector of the (possibly affine) subspace containing $F$ satisfies each relation in $r$.

    Let the symmetric group $S_d$ act on $\mathbb{R}^{d}$ in the standard way. By \Cref{lem:transposition}, any edge $(y_1, y_2)$ of $F$ corresponds to some coordinate transposition $\sigma$ with $\sigma(a) = b, \sigma(b) = a$. Then since $y_1$ and $y_2$ satisfy all the relations in $R$, $y_{1_a} = y_{2_b} \neq y_{2_a} = y_{1_b}$ , and $y_{1_c} = y_{2_c}$ for $c \neq a, b$, it follows that $v_{i_a} = v_{i_b}$ for all $1 \leq i \leq d-k$. This means that for any $y \in V_F$, $\sigma(y) \in V_F$, i.e., $\sigma$ fixes $V_F$.
    
    Define graph $g_F$ with vertices $[d]$ and, for each edge of $F$, define an edge in $g_F$ between the pair of coordinates transposed by its corresponding $\sigma$. Edges in $F$ therefore correspond to (adjacent) value transpositions, and edges in $g_F$ correspond to (possibly non-adjacent) coordinate transpositions; for example, $(y_1, y_2)$ above would yield an edge $(a,b)$ in $g_F$. We can group the edges of $F$ into equivalence classes where two edges are equivalent if and only if they belong to the same connected component in $g_F$. Say the connected components of $g_F$ are $B = \{B_1,...,B_n\}$, where the $B_i$ partition $[d]$. We begin decomposing $F$ in the following claim.
    
    \begin{claim}{3.25.1}
    \label{claim:G_F}
        Let $G_F$ be the set of permutations such that fix the vertices of $F$, i.e., $\sigma(V_F) = V_F$ for all $\sigma \in G_{F}$. Then $G_F$ is a subgroup of $S_d$, and it admits the group direct product decomposition $G_F = \prod_{j=1}^{n}S_{|B_i|}$. 
    \end{claim}
    \begin{proof}
        For any $\sigma \in G_{F}$, we see that $\sigma^{-1} = \sigma^{\ord(\sigma) - 1}$, where $\ord$ denotes group element order. Since powers of $\sigma$ fix $V_F$, $\sigma^{-1} \in G_{F}$. Clearly, the identity is in $G_F$. If $\sigma_1,\sigma_2 \in G_F$ then $\sigma_1(\sigma_{2}(V_F)) = \sigma_1(V_F) = V_F$, so $\sigma_{1}\sigma_{2} \in G_F$. It follows that $G_F$ is a subgroup.
        
        For each $B_i$, $G_F$ contains a collection of coordinate transpositions that form a spanning tree $t_{B_i}$. We show that the these coordinate transpositions generate the subgroup $S_{|B_i|} \subset S_d$ acting on the coordinates of $B_i = \{i_1,...,i_{|B_i|}\}$. Let $i_j$ be a vertex of $t_{B_i}$, and let $\sigma \in S_{|B_i|}$ transpose $i_j$ and some $i_k$. Let $p = (\sigma_{1},\sigma_{2},...,\sigma_{q})$ be a path of edges from $i_j$ to $i_k$. Then $\sigma = (\sigma_{1} \sigma_{2} ... \sigma_{q-1}) \sigma_{q} (\sigma_{q-1} \sigma_{q-2} ... \sigma_{1})$. Since the set of all transpositions in $S_{|B_i|}$ generates $S_{|B_i|}$, so do the transpositions in $t_{B_i}$. Moreover, since every edge of $t_{B_i}$ fixes $V_F$, and the edges of $t_{B_i}$ generate $S_{|B_i|}$ then every $\sigma \in S_{|B_i|}$ fixes $V_F$. This yields the group direct product decomposition $G_F = \prod_{j=1}^{n}S_{|B_i|}$. 
    \end{proof}

    The set of values in $\{0,1,...,d-1\}$ that appear at coordinates in $B_i$ must be a contiguous range of integers, denoted $R_i$, because by \Cref{lem:transposition} all edges of $F$ switch two (possibly non-neighboring) coordinates with neighboring values. Let $T_i$ be the $i$th largest range in $R = \{R_1, \ldots, R_n\}$. Relabel the indices of $B$ so that $B_i$ corresponds to the range $T_i$. Since $S_{|B_i|}$ fixes the coordinates of $B_i$, recalling the definition of $I_J$ from \Cref{lem:simple_face_decomposition}, $F$ restricted to the coordinates in $B_i$ is $CH(P_{B_i}) + \min(T_i) I_{B_i}$, and $F$ has direct sum decomposition $\oplus_{i=1}^{n}[CH(P_{B_i}) + \min(T_i) I_{B_i}]$. Since $CH(P_{B_i}) + \min(T_i) I_{B_i}$ has dimension $|B_i| - 1$, $F$ has dimension $\sum_{i=1}^{n}(|B_i| - 1) = d - n$. Because $F$ has dimension $k$, $n = d - k$.
    
    We have shown that every $k$-dim face of $CH(P_d)$ corresponds to a sequence of subsets $B_1,...,B_{d-k}$ that partition $[d]$. Next, we will complete the claimed bijection by showing the converse. Let $B_1,...,B_{d-k}$ be a sequence of subsets partitioning $[d]$. Let $v$ be the vector with the values of $T_i$ at the coordinates of $B_i$ in any order. Then define $V_F$ to be the orbit of $v$ under the group action $\prod_{i=1}^{d-k}S_{|B_i|}$. Then $CH(V_F) = \prod_{i=1}^{d-k}(CH(P_{B_i}) + \min(T_i)I_{B_i})$ and since $\dim(CH(P_{B_i})) = |B_i| - 1$ then $\dim(CH(V_F)) = \sum_{i=1}^{d-k}(|B_i| - 1) = d - (d - k) = k$. It remains to show that $CH(V_F)$ lies on the boundary of $CH(P_d)$. Let $C_i = \cup_{j=1}^{i}B_i$ and let $U_i = \sum_{j=1}^{i}\sum_{y \in T_j} y$. For $1 \leq i \leq d-k$, let $r_{i}$ be the relation $x \cdot I_{C_i} = U_i$. First, any point of $CH(V_F)$ satisfies all these relations by the bilinearity of the $\cdot$ operator since any vertex of $V_F$ satisfies all these relations. Second, any vertex $w \in P_d$ will have that for all $i$, $0 \leq w \cdot I_{C_i} \leq U_i$, because the $(d-k)$ relations $r_1,...,r_{d-k}$ can only be satisfied by vectors where the top $|C_i|$ elements of $[d]$ appear at the coordinates of $C_i$ for all $i$. By the bilinearity of the $\cdot$ operator, this statement is also true for any point $w \in CH(P_d)$ since it is a convex combination of points in $P_d$. Define the continuous linear functional $f(x) = x \cdot (\sum_{i=1}^{d-k}I_{C_i})$. As $CH(P_d)$ is compact, $f$ is bounded on $CH(P_d)$. The points in $CH(V_F)$ maximize $f$ attaining the value $\sum_{i=1}^{d-k}U_i$. But $f$ cannot attain a maximum value on the interior of $CH(P_d)$ because if it did, say at point $p$, then we can slightly shift $p$ in the direction of $I_{C_i}$ while staying in the interior of $CH(P_d)$, which would increase the value of $f$. It follows that $CH(V_F)$ is on the boundary of $CH(P_d)$. 
\end{proof}

\volumeCalculation*
\begin{proof}
    By \Cref{lem:face_decomposition}, we can write $F$ as $(CH(P_{B_1}) + (\min{T_1})I_{B_1}) \oplus CH(P_{B_2})$. The $(d-2)$-dimensional faces of $CH(P_d)$ with this decomposition are exactly the faces with first subset having size $|B_1|$ and second subset having size $|B_2|$, so there are $\binom{d}{|B_1|}$ faces congruent to $F$.
    
    Turning to volume, $F$ has $(d-2)$-volume
    \begin{equation*}
        |CH(P_{B_1}) + (\min{T_1})I_{B_1}||CH(P_{B_2})| = |CH(P_{B_1})||CH(P_{B_2})|.
    \end{equation*}    
    Previous work has established that $|CH(P_d)| = d^{d-2}V_\diamondsuit$~\citep{ASV21, S86}, where $V_\diamondsuit$ is the volume of the primitive paralleletope $\diamondsuit$ of the lattice $L = \mathbb{Z}^{d} \cap H$ where $H$ is the hyperplane $x_1 + ... + x_d = \frac{d(d-1)}{2}$. It remains to calculate $V_\diamondsuit$.
    
    \begin{claim}
    \label{claim:parallelotope}
        $V_\diamondsuit = \sqrt{d}$.
    \end{claim}
    \begin{proof}
        A primitive parallelotope of a lattice is a minimal collection of vectors that generates the lattice under addition. Pick any point of the lattice to be the origin. Any of the origin's closest neighbors in $L$ is reached from the origin by adding 1 to one coordinate and subtracting 1 from a different coordinate, as this preserves the sum of points required to stay in $H$. For $1 \leq i \leq d-1$, let $v_i$ consist of zeros with 1 at coordinate $i$ and $-1$ at coordinate $i+1$. Then $\{v_1, \ldots, v_{d-1}\}$ generates $L$, so we compute the volume of the resulting parallelotope.
        
        We use the general fact that the $m$-volume of an $m$-parallelotope embedded in $\mathbb{R}^n$ for $n \geq m$ is given by the square root of its Gram determinant, where the Gram determinant of a set of vectors $v_1, \ldots, v_m$ is the determinant of Gram matrix $M$, defined by $M_{i,j} = \langle v_i, v_j \rangle$. The Gram matrix $M_\diamondsuit$ associated with $\diamondsuit$ is a $(d-1)\times(d-1)$ matrix with 2s on the diagonal, -1s on the superdiagonal and subdiagonal, and 0s elsewhere.
        
        We show that $\det(M_\diamondsuit) = d$ by induction on $d$. We apply determinant expansion by minors. For $i \in [d]$, let $M_{\diamondsuit, \neg i}$ denote the $(d-1-i) \times (d-1-i)$ matrix consisting of $M_\diamondsuit$ the last $d-i$ rows and columns of $M_\diamondsuit$. Similarly, let $M_{\diamondsuit, \neg ij}$ denote $M_\diamondsuit$ with row $i$ and column $j$ removed. Applying expansion by minors twice, we get
        \begin{align*}
            \det(M_\diamondsuit) =&\ \sum_{j=1}^d (-1)^{1+j} M_{\diamondsuit, 1j} \det(M_{\diamondsuit,\neg 1j}) \\
            =&\ 2\det(M_{\diamondsuit, \neg 1}) + \det(M_{\diamondsuit, \neg 12}) \\
            =&\ 2\det(M_{\diamondsuit, \neg 1}) - \det(M_{\diamondsuit, \neg 2}) 
        \end{align*}
        Then by the inductive hypothesis, $\det(M_\diamondsuit) = 2(d-1) - (d-2) = d$. The base case $d=3$ has a $2 \times 2$ Gram matrix with determinant $2 \cdot 2 - (-1)(-1) = 3$.
    \end{proof}
    Thus $|CH(P_d)| = d^{d-2}V_{\diamondsuit} = d^{d-3/2}$. It follows that $F$ has volume $|B_1|^{|B_1|-3/2}|B_2|^{|B_2|-3/2}$. 
\end{proof}

\altitudeCalculation*
\begin{proof}
    First, $c(CH(P_d)) = \frac{d-1}{2}I_{[d]}$. Second,
    \begin{align*}
        c(F) =&\ c(CH(P_{B_1}) + (\min{T_1})I_{B_1}) + c(CH(P_{B_2})) \\
        =&\  \left(\frac{|B_1|-1}{2}\right)I_{B_1} + (d-|B_1|)I_{B_1} + \left(\frac{|B_2|-1}{2}\right)I_{B_2} \\
        =&\  \left(\frac{2d - |B_1|-1}{2}\right)I_{B_1} + \left(\frac{|B_2|-1}{2}\right)I_{B_2}.
    \end{align*}
    Let $w$ be the vector pointing from $c(CH(P_d))$ to $c(F)$. Then
    \begin{align*}
        w =&\ c(F) - c(CH(P_d)) \\
        =&\ \left(\frac{d-|B_1|}{2}\right)I_{B_1} + \left(\frac{|B_2|-d}{2}\right)I_{B_2} \\
        =&\ \left(\frac{|B_2|}{2}\right)I_{B_1} - \left(\frac{|B_1|}{2}\right)I_{B_2},
    \end{align*}
    and the length of $w$ is  $\frac{1}{2}\sqrt{|B_1||B_2|^{2} + |B_2||B_1|^{2}}$.

    Let $F_{B_1}$ be $F$ restricted to the coordinates in $B_1$. Write $B_1 = \{i_1,...,i_{|B_1|}\}$ in increasing order. For $1 \leq j \leq |B_1|-1$ let $v_{j} \in \mathbb{R}^{d}$ be the vector 1 at coordinate $i_j$, $-1$ at coordinate $i_{j+1}$, and 0 elsewhere. $V = \{v_1,...,v_{|B_1|-1}\}$ is linearly independent. Moreover, each $v_j$ is equal to a difference of adjacent vertices of $F_{B_1}$, so $v_j$ lies in the same $(|B_1|-1)$-dimensional subspace as $F_{B_1}$. It follows that $V$ is a basis for this subspace. Next, $v_j \cdot w = 0$ for all $j$, so $w$ is orthogonal to $F_{B_1}$. Similarly, $w$ is orthogonal to $F_{B_2}$, so $w$ is orthogonal to $F$.
\end{proof}

\simplexProductTriangulation*
\begin{proof}
    First, we will change basis so that $\Delta_x$ and $\Delta_y$ can be viewed as fundamental simplices (\Cref{def:fundamental_simplex}). Define the sequences $B_x = (x_{n} - x_{n-1},...,x_{2} - x_{1}, x_{1})$ and $B_y = (y_{m} - y_{m-1},...,y_{2} - y_{1}, y_{1})$. Then for $1 \leq i \leq n$ we can write $x_i$ as the sum of the last $i$ vectors in $B_x$. Equivalently, $x_i$ can be written in the $B_x$ basis as the vector which starts with $n-i$ zeros, is followed by $i$ ones, and ends with $m$ zeros, i.e. $\Delta_x$ is a fundamental simplex embedded in $V_x$. Similarly, we can write $\Delta_y$ in the $B_y$ basis as a fundamental simplex embedded in $V_y$. Then any point $p \in \Delta_x \oplus \Delta_y$ in the $B_x, B_y$ bases takes the form $p = (a_1,...,a_n) \oplus (b_{1},...,b_{m})$ where $a_i < a_{i+1}$ and $b_{i} < b_{i+1}$ for all $i$ (\Cref{lem:fundamental_simplex_increasing}).
    
    Note that when we write the direct sum $\Delta_x \oplus \Delta_y$, we are technically talking about an internal direct sum, so we can equivalently represent $p = (a_1,...,a_n) \oplus (b_{1},...,b_{m}) \in \Delta_x \oplus \Delta_y$ as $(a_1,...,a_n,b_1,...,b_m) \in \mathbb{R}^{n+m}$ in the ambient space. In the remainder of the proof, we will use the first representation of $p$ when we want to emphasize which coordinates of $p$ belong to each of $\Delta_x$ and $\Delta_y$, and we will use the second representation when we need to consider the relationship between all the coordinates together. Moreover, we can assume that $\{a_1,...,a_n,b_1,...,b_m\}$ contains no duplicates by excluding a set of points of measure zero, as in \Cref{lem:assumption}.
    
    We want to determine an equivalence relation on the points of $\Delta_x \oplus \Delta_y$ that will decompose it into equal volume simplices. Given $p = (a_1, \ldots, a_n) \oplus (b_1, \ldots, b_m)$ as in the preceding paragraph, let $p'$ be $(a_1, \ldots, a_n, b_1, \ldots, b_m)$ sorted in decreasing order. We define the type vector of $p$ to be the following vector in $\{t_x, t_y\}^{n+m}$: the $i^{th}$ position of the type vector of $p$ is $t_x$ if $p_i' = a_j$ for some $j$, and $t_y$ if $p_i' = b_j$ for some $j$.
    
    Similarly, the $n+m+1$ vertices of any $\Delta \in D$ can be written as $\{x_{f(0)} \oplus y_{g(0)},...,x_{f(n+m)}\oplus y_{g(n+m)}\}$, where $f$ and $g$ denote some interleaving of the form described in the lemma statement, so we define the type vector of $\Delta$: the $i^{th}$ position of the type vector of $\Delta$ is $t_x$ if $f(i) > f(i-1)$, and type $t_y$ if $g(i) > g(i-1)$. We use the following result.
    
    \begin{claim}
    \label{claim:type}
        Let $p \in \Delta_x \oplus \Delta_y$ and $\Delta \in D$. Then $p \in \Delta$ if and only if $p$ and $\Delta$ have the same type vectors.
    \end{claim}
    \begin{proof}
        We can view the vertices $\{x_{f(0)} \oplus y_{g(0)},...,x_{f(n+m)}\oplus y_{g(n+m)}\}$ of $\Delta$ as being iteratively constructed from left to right as follows. In the $B_x, B_y$ basis, vertex $x_{f(0)} \oplus y_{g(0)} = x_0 \oplus y_0 \in \mathbb{R}^{n}\oplus \mathbb{R}^{m}$ is written as $(0,...,0)\oplus(0,...,0)$. For $i > 0$, each subsequent vertex $x_{f(i)} \oplus y_{g(i)}$ is formed from the previous vertex $x_{f(i-1)} \oplus y_{g(i-1)}$ by first picking either the subvector corresponding to $\Delta_x$ (the first $n$ coordinates) or the subvector corresponding to $\Delta_y$ (the last $m$ coordinates), and then replacing the rightmost 0 by 1 in that subvector. For $i \in [n+m]$, define $h(i) \in [n+m]$ to be the coordinate that is replaced at step $i$ in the iterative construction of the vertices of $\Delta$. Define the support $S(h(i))$ of $h(i)$ to be the subset of the vertices of $\Delta$ where the value at $h(i)$ is 1. Then $S(h(1)) \supset S(h(2)) \supset ... \supset S(h(n+m))$.
            
        Any $p \in \Delta$ is some convex combination of the vertices of $\Delta$, so in the $B_x, B_y$ bases $p_{h(1)} \geq p_{h(2)} \geq ... \geq p_{h(n+m)}$. Let $t_{p}$ be the type vector of $p$, and let $t_{\Delta}$ be the type vector of $\Delta$. If $h(i) \in [n]$ then $t_{p_i} = t_x$ by the chain of inequalities above and $t_{\Delta_i} = t_x$ since the replacement of the rightmost 0 by 1 in the subvector corresponding to $\Delta_x$ at step $i$ is equivalent to $f(i) > f(i-1)$. Similarly, if $h(i) \in \{n+1,...,n+m\}$ then $t_{p_i} = t_y = t_{\Delta_i}$. So $t_{p} = t_{\Delta}$ for all $p \in \Delta$.
            
        Conversely, given any point $p \in \Delta_x \oplus \Delta_y$, let $\Delta \in D$ be the simplex with the same type vector as $p$. As before, we can write $p$ in the $B_x, B_y$ bases as $(a_1,...,a_n) \oplus (b_{1},...,b_{m})$ and sorted in descending order as $(p_1',...,p_{n+m}')$, and write the vertices of $\Delta$ as $\{x_{f(0)} \oplus y_{g(0)},...,x_{f(n+m)}\oplus y_{g(n+m)}\}$. Note that $0 \leq a_i \leq 1$ and $0 \leq b_i \leq 1$ for all $i$ since $\Delta_x$ and $\Delta_y$ are fundamental simplices in the $B_x,B_y$ basis, so $0 \leq p_{i}' \leq 1$ for all $i$. Define $d_0 = 1 - p_{1}', d_{n+m} = p_{n+m}'$, and for $1 \leq i \leq n+m-1$ define $d_i = p_{i}' - p_{i+1}'$. Since $0 \leq p_{i}' \leq 1$ for all $i$ and the $p_{i}$'s are descending, $0 \leq d_{j} \leq 1$ for all $j$. Then $p = \sum_{j=0}^{n+m}d_{j}(x_{f(j)}\oplus y_{g(j)})$ and since $\sum_{j=0}^{n+m}d_{j} = 1$ then $p$ is a convex combination of vertices of $\Delta$, so $p \in \Delta$.
    \end{proof}

    It follows that $D$ decomposes $\Delta_x \oplus \Delta_y$ into simplices. For any simplex in $D$, if we consider the matrix whose rows are its vertices, there is some permutation of its columns such that the resulting matrix's rows are the vertices of the fundamental simplex in $\mathbb{R}^{n+m}$, so every simplex in $D$ has the same volume.
\end{proof}

\begin{lemma}
\label{lem:uniform_CHPD}
    We can sample a point uniformly at random from $CH(P_d)$ in time $O(d^2\log(d))$.
\end{lemma}
\begin{proof}
    First partition $CH(P_d)$ into pyramids whose bases are the $(d-2)$-dimensional faces and whose shared apex is $c(CH(P_d))$. By \Cref{lem:volume calculation} and \Cref{lem:altitude calculation} we know the $(d-2)$-volume $A$  of each base, their multiplicity, and the height of each altitude $h$, so we can sample a pyramid with weight proportional to its volume $\frac{Ah}{d}$.
    
    Explicitly, define equivalence classes of $(d-2)$-dimensional faces $\{F_1,...,F_{d-1}\}$ partitioned by congruence. Specifically, $F_j$ is the set of faces corresponding to a sequence of subsets $B_1, B_2$ with $|B_1| = j, |B_2| = d-j$. Then assign weight $w_j = M_jV_jH_j$ to each equivalence class $F_j$ where $M_j = \binom{d}{j}, V_j = j^{j-3/2}(d-j)^{(d-j)-3/2}, H_j = \frac{1}{2}\sqrt{(j)(d-j)^{2} + (d-j)j^{2}}$. Sample a class $F_j$ according to $w_j$. Then sample a particular member $F \in F_j$ by first drawing a random permutation $\sigma$ of $[d]$ and then setting $B_{1}$ to be the first $i$ elements of $\sigma$, and assigning $B_2 = [d] - B_1$, as in \Cref{lem:face_decomposition}. Then $F$ has direct sum decomposition $(CH(P_{B_1}) + (\min{T_1})I_{B_1}) \oplus CH(P_{B_2})$.

    Having sampled a pyramid, the next step is to decompose the pyramid into simplices. Recursively sample a simplex $\Delta_1$ with the appropriate probability from a star decomposition of $CH(P_{B_1})$ and sample a simplex $\Delta_2$ with the appropriate probability from a star decomposition of $CH(P_{B_2})$. By \Cref{lem:simplex_product_triangulation}, we can decompose $(\Delta_1 + (d-|B_1|)I_{B_1}) \oplus \Delta_2$ into equal volume simplices that are indexed by a type vector in $\{t_{\Delta_1}, t_{\Delta_2}\}^{|B_1|+|B_2|-2}$ where $t_{\Delta_1}$ appears $|B_1|-1$ times and $t_{\Delta_2}$ appears $|B_2|-1$ times. Uniformly sample a simplex $\Delta_3 \in (\Delta_1 + (d-|B_1|)I_{B_1} \oplus \Delta_2$. Then the pyramid $K$ with base $\Delta_3$ and apex $c(CH(P_d))$ is a simplex sampled from a star decomposition of $CH(P_d)$ with the appropriate probability. In the base case of $d=1$, a star decomposition of the single point set $CH(P_1) = P_1$ is itself, so we just return the point. To sample a point uniformly at random from $CH(P_d)$, we return a point uniformly sampled from the simplex $K$. 
    
    We now consider running time. Pre-computing all possible values of $\binom{d}{i}$ takes time $O(d^2)$. For each iteration, computing the $M_i$ (from the pre-computed binomials) and $H_i$ takes constant time. For $V_i$, it suffices to consider the time to compute $d^d$. Consider the binary expansion $d = b_0 + 2b_1 + 4b_2 + \ldots + 2^kb_k$ for bits $b_i$ and $k = \lfloor \log(d) \rfloor$. Then we can compute each successive $d^{b_0}, \ldots, d^{2^kb_k}$ using the previously computed term with a nonzero $b_i$ in a single pass of time $O(\log(d))$, and multiplying them together to compute $d^d$ takes another $\log(d)$ operations. It therefore takes $O(d\log(d))$ time overall to compute $\{w_1, \ldots, w_d\}$. Drawing a random permutation takes time $O(d)$, and this suffices to sample a pyramid. Once we have sampled the pair of simplices $\Delta_1$ and $\Delta_2$, uniformly sampling $\Delta_3$ corresponds to uniformly sampling a type vector in $\{t_{\Delta_1}, t_{\Delta_2}\}^{|B_1|+|B_2|-2}$ where $t_{\Delta_1}$ appears $|B_1|-1$ times and $t_{\Delta_2}$ appears $|B_2|-1$ times which costs the time it takes to pick a subset of $|B_1|-1$ indices from $[|B_1| + |B_2| - 2]$, which we can do by picking a random permutation of $[|B_1| + |B_2| - 2]$ and then picking the first $|B_1|-1$ indices, which costs $O(d)$. We recurse $O(d)$ times, so the overall time is $O(d^2\log(d))$.
\end{proof}

\begin{theorem}
\label{thm:vote}
    We can sample a point uniformly at random from the cylinder $C$ with bases $CH(P_d)$ and $-CH(P_d)$ in time $O(d^2\log(d))$.
\end{theorem}
\begin{proof}
    As $-CH(P_d)$ is a reflection of $CH(P_d)$ across the hyperplane $x_1 + ... + x_d = 0$, the distance between a point $p \in CH(P_d)$ and its reflection $p' \in -CH(P_d)$ is constant. Explicitly, $p' = p - (d-1)I_{[d]}$. To sample uniformly from $C$, it suffices to uniformly sample a point $p \in CH(P_d)$, which we can do by \Cref{lem:uniform_CHPD}, and then uniformly sample a point on the line segment joining $p$ to $p'$. 
\end{proof}

\begin{algorithm}
    \begin{algorithmic}[1]
    \STATE {\bfseries Input:} Dimension $d$
    \IF{d = 1}
        \STATE return $\{(0)\}$
    \ENDIF
    \FOR{$j=1, \ldots, d-1$}
        \STATE Compute permutohedron face class weight $w_j = M_{j}V_{j}H_{j}$ as in the proof of \Cref{lem:uniform_CHPD}
    \ENDFOR
    \STATE Sample face class $j \propto w_j$
    \STATE Uniformly sample a random permutation $\sigma$ of $[d]$
    \STATE Let $B_1$ be the first $j$ elements of $\sigma$ and let $B_2 = [d] - B_1$
    \STATE Recursively call Algorithm 2 with input $|B_1|$ to sample $(j-1)$-simplex $\Delta_1 \in CH(P_{B_1})$
    \STATE Recursively call Algorithm 2 with input $|B_2|$ to sample $(d-j-1)$-simplex $\Delta_2 \in CH(P_{B_2})$
    \STATE Uniformly sample type vector $t$ in $\{t_{\Delta_1}, t_{\Delta_2}\}^{d-2}$ with $j-1$ instances of $t_{\Delta_1}$ and $d-j-1$ instances of $t_{\Delta_2}$
    \STATE Compute $(d-2)$-simplex $\Delta_3 \in (\Delta_1 + (d-j)I_{B_1}) \oplus \Delta_2$ corresponding to type vector $t$ as in \Cref{lem:simplex_product_triangulation} 
    \STATE Let $K$ be the $(d-1)$-simplex formed by appending $c(CH(P_d))$ to the list of vertices of $\Delta_3$
    \STATE Return $K$
    \caption{Vote Sampler}
    \label{alg:vote_sampler}
    \end{algorithmic}
\end{algorithm}

\subsection{Proofs For Vote Rejection Sampling}

\VoteEnclosing*
\begin{proof}
    The maximum $\ell_p$ norm of a point in $CH(P_d)$ is achieved at any of the vertices, which are permutations of $(0,1,\ldots, d-1)$. These have $\ell_p$ norm $\left(\sum_{j=0}^{d-1} j^p\right)^{1/p}$ for $p \in [1,\infty)$ and $\ell_\infty$ norm $d-1$.
\end{proof}

\voteRejectionSample*
\begin{proof}
    Recall from the analysis of \Cref{lem:volume calculation} that the cylinder $V$ has base $(d-1)$-volume $d^{d-3/2}$ and height $(d-1)\sqrt{d}$, for a total volume upper bounded by $d^d$. We split into two cases, depending on the relationship between $d$ and $p$, and show in each that the enclosing $\ell_p$ ball volume is much larger than $d^d$. Both cases start by applying \Cref{lem:lp_volume} and \Cref{lem:vote_enclosing} to get
    \begin{equation*}
        V_p^d = \frac{2^d\left(\sum_{j=0}^{d-1} j^p\right)^{d/p}\Gamma(1+\frac{1}{p})^d}{\Gamma(1+\frac{d}{p})}.
    \end{equation*}
    
    \underline{Case 1}: $d \leq p$. Then by \Cref{lem:lp_volume} and \Cref{lem:vote_enclosing},
    \begin{align*}
        V_p^d >&\ 2^d \cdot (d-1)^d \cdot 0.885^d \\
        =&\ (1.77(d-1))^d
    \end{align*}
    where the inequality uses the fact that $\Gamma(x) \geq 0.885$ and $0 < \Gamma(1+\frac{d}{p}) < 1$. The minimum enclosing $\ell_\infty$ ball has volume $(2(d-1))^d$. Note that $\frac{V}{V_{p}^{d}} \leq  1.77^{-d}(\frac{d}{d-1})^{d} \leq 4(1.77^{-d})$ where we have used that $(\frac{d}{d-1})^{d}$ is monotonically decreasing. Then it takes at least an expected $\frac{(1.77)^{d}}{4}$ samples to hit a success.
    
    \underline{Case 2}: $d > p$. Consider the Riemann sum
    \begin{align*}
        \lim_{d\rightarrow \infty} \frac{1}{d}\sum_{j=0}^{d-1}\left(\frac{j}{d}\right)^{p} = \int_{0}^{1} x^{p} dx = \frac{1}{p+1}
    \end{align*}
    
    Define 
    \begin{equation*}
        U = \frac{1}{d}\sum_{j=1}^{d}\left(\frac{j}{d}\right)^{p} \text{ and }
        L = \frac{1}{d}\sum_{j=0}^{d-1}\left(\frac{j}{d}\right)^{p} \text{ and }
        I = \int_{0}^{1} x^{p} dx = \frac{1}{p+1}.
    \end{equation*}
    
    Since $f(x) = x^{p}$ is convex on $x \in [0,1]$, the trapezoidal sum is an upper bound for the integral, i.e. $\frac{1}{2}(L + U) \geq I$. We also have $\frac{1}{2}(L - U) = -\frac{1}{2d}$. Summing the inequality and the equation, we get $L \geq I -\frac{1}{2d} = \frac{1}{p+1} - \frac{1}{2d} \geq \frac{1}{2(p+1)}$. This gives the lower bound $\sum_{j=0}^{d-1}j^{p} \geq d^{p+1}/[2(p+1)]$.
    
    We use the following bounds to analyze the $\Gamma$ terms in $V_p^d$.
    
    \begin{claim}[\cite{N23}]
        Let $x > 0$ and $\alpha = \sqrt{2\pi} \cdot x^{x-1/2}e^{-x}$. Then 
        \begin{equation*}
        \alpha < \Gamma(x) < \alpha \cdot \exp\left(\frac{1}{12x}\right).
        \end{equation*}
    \end{claim}

    Then by $d > p$,
    \begin{align*}
        \Gamma\left(1 + \frac{d}{p}\right) &= \frac{d}{p}\Gamma\left(\frac{d}{p}\right) \\ 
        &< e^{1/12}\left(\frac{d}{p}\right)\left(\frac{ p}{d}\right)^{1/2}\left(\frac{d}{pe}\right)^{d/p} \\
        &\leq e^{1/12}\left(\frac{d}{p}\right)\left(\frac{d}{pe}\right)^{d/p}.
    \end{align*}
    
    Finally, we lower bound $V_p^d$.
    \begin{align*}
        V_p^d &= \frac{2^d\left(\sum_{j=0}^{d-1} j^p\right)^{d/p}\Gamma(1+\frac{1}{p})^d}{\Gamma(1+\frac{d}{p})} \\
        &\geq 
        \frac{2^d\left[\left(\frac{1}{2(p+1)}\right)d^{p+1}\right]^{d/p}\Gamma(1+\frac{1}{p})^d}{\Gamma(1+\frac{d}{p})} \\
        &\geq 
        \frac{2^d\left[\left(\frac{1}{2(p+1)}\right)d^{p+1}\right]^{d/p}\Gamma(1+\frac{1}{p})^d}{e^{1/12}\left(\frac{d}{p}\right)\left(\frac{d}{pe}\right)^{\frac{d}{p}}}
    \end{align*}
    by our lower bound on $\sum_{j=0}^{d-1} j^p$ and upper bound on $\Gamma(1+\tfrac{d}{p})$, respectively. We continue 
    \begin{align*}
        \frac{2^d\left[\left(\frac{1}{2(p+1)}\right)d^{p+1}\right]^{d/p}\Gamma(1+\frac{1}{p})^d}{e^{1/12}\left(\frac{d}{p}\right)\left(\frac{d}{pe}\right)^{\frac{d}{p}}} &\geq 
        e^{-1/12}2^{d}d^{d}\left(\frac{p}{d}\right)(pe)^{d/p}\left(\frac{1}{2(p+1)}\right)^{d/p}\Gamma\left(1+\frac{1}{p}\right)^d \\
        &\geq 
        e^{-1/12}2^{d}d^{d}\left(\frac{p}{d}\right)e^{d/p}\left(\frac{p}{2(p+1)}\right)^{d/p}\Gamma\left(1+\frac{1}{p}\right)^d \\
        &\geq
        e^{-1/12}2^{d}d^{d}\left(\frac{p}{d}\right)e^{d/p}\left(\frac{1}{4}\right)^{d/p}(0.885)^d \\
        &\geq
        e^{-1/12}2^{d}d^{d}\left(\frac{p}{d}\right)(0.679)^{d/p}(0.885)^d \\
        &\geq
        e^{-1/12}2^{d}d^{d}\left(\frac{1}{d}\right)(0.679)^{d}(0.885)^d \\
        &\geq
        e^{-1/12}(1.2)^{d}d^{d-1} \\
        &\geq
        (1.2d)^{d-1}
    \end{align*}
    Then $\frac{V}{V_{p}^{d}} \leq d(1.2)^{-d+1}$ so that it takes an expected $\frac{(1.2)^{d-1}}{d}$ samples before hitting a success.
\end{proof}

\subsection{Proofs For Vote Ellipse}
Since proofs that the minimum ellipse of $\voteball$ is origin-centered, unique, and has the same axis directions as the minimum ellipse of $\countball$, we need only solve its program.
\voteEllipseTheorem*
\begin{proof}
    Let $e_{w_1} = \frac{1}{\sqrt{d}}(1,...,1)$, and let $e_{w_2} = \frac{1}{\sqrt{2}}(-1, 1, 0,...,0)$. Extend $\{e_{w_1}, e_{w_2}\}$ to a full orthonormal basis $B = \{e_{w_1},...,e_{w_d}\}$. Define $w_{1} = (\frac{d-1}{2},...,\frac{d-1}{2})$ and
    \begin{equation*}
        w_2 = (0,1,...,d-1) - w_1 = \left(-\frac{d-1}{2}, -\frac{d-3}{2},...,\frac{d-1}{2}\right)
    \end{equation*}
    so $\|w_{1}\|_2 = \frac{(d-1)\sqrt{d}}{2}$ and
    \begin{align*}
        \|w_2\|_2 =&\ \sqrt{\sum_{i=0}^{d-1} \left[i - \frac{d-1}{2}\right]^2} \\
        =&\ \sqrt{\sum_{i=0}^{d-1} \left[i^2 - i(d-1) + \frac{(d-1)^2}{4}\right]} \\
        =&\ \sqrt{\sum_{i=0}^{d-1} i^2 - (d-1)\sum_{i=0}^{d-1} i + \frac{d(d-1)^2}{4}} \\
        =&\ \sqrt{\frac{d(d-1)(2d-1)}{6} - \frac{d(d-1)^2}{2} + \frac{d(d-1)^2}{4}} \\
        =&\ \sqrt{\frac{d(d-1)}{12} \cdot [2(2d-1) - 6(d-1) +  3(d-1)]} \\
        =&\ \sqrt{\frac{d(d^2-1)}{12}}.
    \end{align*}
    Note that \Cref{lem:centering_lemma} and \Cref{lem:axes_directions}  hold for the cylinder of $CH(P_d)$ because it is symmetric about its center and contains all the symmetries that $T^{d}$ contains.
    We can rotate the ellipse so that it intersects $\|w_1\|_2e_{w_1} + \|w_2\|_2e_{w_2}$ since every vertex of $CH(P_d)$ contacts $E$. This point is written as $(\|w_1\|_2,\|w_2\|_2,0,...,0)$ in the $B$ basis.
    
    Consider the program whose objective function is $f(a_1, a_2) = a_{1}^{2} + (d-1)a_{2}^{2}$, and whose constraint in the $B$ basis can be written as $g(a_1, a_2) = \frac{\|w_1\|_2^{2}}{a_{1}^{2}} + \frac{\|w_2\|_2^{2}}{a_{2}^{2}} - 1 = 0$. This program can be solved via Lagrange multipliers and there is a unique solution. 
    
    Define the Lagrangian $\mathcal{L}(a_1, a_2, \lambda) = f(a_1, a_2) + \lambda g(a_1, a_2)$. Any optimal point of $\mathcal{L}$ satisfies that $\nabla\mathcal{L} = 0$. Following the same calculation as in \Cref{thm:count_ellipse} with $w_i$ in place of $v_i$,
    \begin{align*}
         a_1 &= \left(\lambda\|w_1\|_2^{2}\right)^{1/4} \\
         a_2 &= \left(\frac{\lambda\|w_2\|_2^{2}}{d-1}\right)^{1/4} \\ \lambda &= (\|w_1\|_2 + \|w_2\|_2\sqrt{d-1})^{2}
    \end{align*}
    and expressions for $a_1$ and $a_2$ in terms of $d$ follow by substituting the closed form for $\lambda$.
\end{proof}
\section{Parallelized Elliptic Gaussian Noise}
\label{sec:appendix_parallel}
We want to sample from a random ellipse $RE$ (see \Cref{def:random_ellipse}) in a parallelized manner.
\begin{lemma}
    There is a parallelized algorithm to sample a point uniformly from the random ellipse RE in parallel runtime $O(\log(d))$.
\end{lemma}
\begin{proof}
    Let $W_1,...,W_d$ be parallel workers. Let $M$ be the central manager. In the following pseudo-code, the for loops over the workers are done in parallel.

\begin{algorithm}
    \begin{algorithmic}[1]
    \STATE {\bfseries Input:} Dimension $d$, $\ell_0$ bound $k$, axis lengths $a_1$ and $a_2$ 
    \FOR{$j=1, \ldots, d$}
        \STATE Worker $W_j$ samples $X_j \sim \mathcal{N}(0,1)$
    \ENDFOR
    \STATE Manager $M$ computes $s = \frac{1}{d}\sum_{j=1}^{d}a_{2}X_j$
    \STATE Manager $M$ distributes a copy of $s$ to each worker $W_j$
    \FOR{$j=1, \ldots, d$}
        \STATE Worker $W_j$ computes $Z_j = a_{2}X_j + s(-1 + \frac{a_1}{a_2})$
    \ENDFOR
    \RETURN $Z$
    
    \caption{Parallelized Ellipse Gaussian Noise Sampler}
    \label{alg:parallelized_ellipse_gaussian_noise_sampler}
    \end{algorithmic}
\end{algorithm}

At a high level, the strategy will be to:
\begin{enumerate}
    \item Generate a sample from $\mathcal{N}(0,I_d)$ centered at the origin.
    \item Scale it by the axis length $a_2$. This step will scale all the directions among $\{v_2,...,v_d\}$ correctly but will scale the direction $v_1 = (1,...,1)$ incorrectly.
    \item Correct scaling in the $v_1$ direction.
\end{enumerate}
The rest of the proof verifies that this produces the appropriate $Z$. For step 1, let $X \sim \mathcal{N}(0, I_d) =  (\mathcal{N}(0,1),...,\mathcal{N}(0,1))$. We first let worker $W_j$ generate $X_j \sim \mathcal{N}(0,1)$ in parallel runtime $O(1)$. Write $X = RY$ where $R \sim \chi_{d}$ and $Y$ is a uniform sample of the origin centered unit sphere.

For step 2, each worker $W_j$ will compute $a_{2}X_{j}$. Then $a_{2}X$ is a uniform sample from the distribution $a_{2}RY$. At this point, each of the directions in $\{v_2,...,v_d\}$ have been scaled by $a_2$.

In step 3, the component of $a_{2}X$ in the $v_1 = (1,...,1)$ direction is given by $s = \frac{1}{d}\sum_{j=1}^{d}a_{2}X_{j}$ which can be computed by manager $M$ by a reduce in parallel runtime $O(\log(d))$. The correction to $a_{2}X$ to account for the proper $a_{1}$ length would then be $a_{2}X -s(1,..,1) + s(\frac{a_1}{a_2})(1,...,1)$. To compute this in a parallel way, manager $M$ sends $s$ to each worker $W_j$ who takes the result of step 2 and computes $Z_{j} = a_{2}X_{j} + s(-1 + \frac{a_1}{a_2}) = R\left(a_{2}Y_{j} + \frac{s}{R}(-1 + \frac{a_1}{a_2})\right)$. Since $a_{2}Y_{j} + \frac{s}{R}(-1 + \frac{a_1}{a_2})$ is a uniform sample from $E$, $Z$ is a uniform sample from $RE$.
\end{proof}

\end{document}